\documentclass[10pt,peerreview]{IEEEtran}
\usepackage[dvips]{graphicx}
\usepackage[english]{babel}
\usepackage[latin1]{inputenc}

\usepackage{amssymb,algorithmic,algorithm}
\usepackage{url}

\begin{document}
\newtheorem{thm}{\textbf{Theorem}}
\newtheorem{thm1}{Theorem}
\newtheorem{thm2}{Theorem}
\newtheorem{thm3}{Theorem}
\newtheorem{lemma}[thm1]{\textbf{Lemma}}
\newtheorem{corollary}[thm2]{\textbf{Corollary}}
\newtheorem{proposition}[thm3]{\textbf{Proposition}}
\newtheorem{definition}{\textbf{Definition}}
\newtheorem{remark}{\textbf{Remark}}
\newtheorem{assumption}{\textbf{Assumption}}
\newtheorem{question}{\textbf{Question}}
\newtheorem{answer}{\textbf{Answer}}

%
\title{Optimal Antenna Allocation in MIMO Radars with Collocated Antennas}

\author{A.A. Gorji$^{a}$, R. Tharmarasa$^{a}$, W.D. Blair$^{b}$ and T. Kirubarajan$^{a}$\\
$^a$Department of Electrical and Computer Engineering, McMaster University, Hamilton, Ontario, Canada \\
$^b$Georgia Tech Research Institute, GA, USA}



%


\maketitle

\begin{abstract}
This paper concerns with the sensor management problem in collocated Multiple-Input Multiple-Output (MIMO) radars. After deriving the Cramer-Rao Lower Bound (CRLB) as a performance measure, the antenna allocation problem is formulated as a standard Semi-definite Programming (SDP) for the single-target case. In addition, for multiple unresolved target scenarios, a sampling-based algorithm is proposed to deal with the non-convexity of the cost function. Simulations confirm the superiority of the localization results under the optimal structure.
\end{abstract}

%
\section{Introduction}
Multiple-Input Multiple-Output (MIMO) radars with collocated antennas have been introduced recently in the literature \cite{Li2007}\cite{Li2009} as an alternative to the traditional phased-array radar systems \cite{SkolniK2002}. Unlike the conventional phased-array radar systems in which the transmitted signals are fully-coherent, MIMO radars enjoy the orthogonality of transmitted waveforms. The orthogonality of the transmitted signals provides a number of benefits for the MIMO radars, such as the diversity in the paths \cite{Fishler2006}, virtual aperture extension \cite{Bekkerman2006}, beam pattern improvement \cite{Bekkerman2006}, and higher probability of detection \cite{Bekkerman2006}, over the phased-array systems. Consequently, this has generated much interest among researchers to analyze different aspects of collocated MIMO radars such as waveform selection \cite{Fuhrmann2008}\cite{He2011}\cite{Li2008}\cite{Stoica2008}, range compression \cite{Li2008} and the applications of collocated MIMO radars in target detection, localization and tracking \cite{Gorji2011}\cite{Roberts2012}\cite{Xu2008}.

The location CRLB was recently proposed in the literature as a tool to evaluate the localization performance of collocated MIMO radars \cite{Bekkerman2006}\cite{Gorji2011}\cite{Li2008}. The Direction-of-Arrival (DOA) of the target was defined as the parameter of the problem in \cite{Bekkerman2006}. Then, the CRLB was derived according to the received complex signals. When multiple targets fall inside the same resolution cell of the MIMO radar, the CRLB might be also affected accordingly. An alternate form of the CRLB was then derived in \cite{Li2008} and the effect of the number of targets occupying the same cell might affect on the CRLB was analyzed. While previous works only derived the CRLB for DOA estimation, it was shown in \cite{Gorji2011} that the range information of the target can be also included in the received measurements. Therefore, a novel measurement model was proposed in \cite{Gorji2011} and the CRLB was found for both range and DOA of the target. It was also shown that the CRLB is affected by the number and locations of targets falling inside the same resolution cell.

Antenna allocation is a critical concern in MIMO array systems. An optimal antenna placement algorithm, where an array of closely-spaced antennas received the Time-of-Arrival (TOA) data, was proposed in \cite{Bishop2010}. It is also shown in \cite{Tharmarasa2007} that the Posterior CRLB (PCRLB) \cite{Tichavsky1998} can be used to find the number and optimal locations of multiple sensors while there is no restriction on the closeness of inter-sensor distances. The CRLB was also employed in \cite{Godrich2010}\cite{He2010} for antenna placement in widely-separated MIMO radars. It was shown that the trace of the CRLB matrix can be written as a convex function of the location of antennas. Then, convex optimization techniques were applied in order to find the optimal placement of antennas. The CRLB was also used as a performance metric in \cite{Godrich2012} for the antenna selection in widely-separated MIMO radars where a subset of antennas has to be chosen out of a large number of antennas that are widely-separated in the surveillance region.

Recently, there has been interest in applying optimization techniques to different aspects of collocated MIMO radars as well. The CRLB was employed in \cite{Li2008} to find the optimal cross-correlation matrix of the transmitted signals where it was shown that the CRLB is a convex function of the cross-correlation matrix. A Gradiant-based approach was also formulated in \cite{Fuhrmann2008} for beam-pattern synthesis by optimizing the transmitters' cross-correlation matrix. Although several other efforts have been made on the waveform optimization in collocated MIMO radars \cite{Chen2008}\cite{Chen2009}, the problem of antenna allocation in collocated MIMO radars has not been addressed so far. In \cite{He2010}, an algorithm for antenna selection in collocated MIMO radars was presented. Nevertheless, the proposed technique does not provide any systematic procedure for distributing the antennas in the surveillance region when the optimal set is chosen. It was demonstrated in \cite{Gorji2011} that the CRLB of a collocated MIMO radar is a function of the location of antennas. The simulations in \cite{Gorji2011} also showed that the localization performance is affected by the distribution of antennas in the surveillance region. Consequently, it is of great interest to find an optimal distribution of antennas that provide the best localization performance.

In this paper, the antenna allocation problem for collocated MIMO radar systems is addressed and a systematic approach is proposed based on the CRLB. To the best of our knowledge, there is no comprehensive work on the design and analysis of an optimal antenna placement framework for collocated MIMO radars. The main contributions of this paper are as follows:
\begin{itemize}
\item A novel CRLB derivation for MIMO radars with collocated antennas:\\
Although the CRLB was derived in the literature for the collocated MIMO radars \cite{Bekkerman2006}\cite{Li2008}, the effect of the range information was not considered in the CRLB derivation. In addition, there is no compact CRLB derivation in terms of the location of antennas. In this paper, the CRLB is first derived for a collocated MIMO radar where both DOA and range of the target are both embedded in the signal model. Also, the impact of the situation in which multiple targets fall inside the same resolution cell is taken into consideration.
\item A convex optimization approach for the single-target case:\\
It is shown that the antenna allocation problem can be dealt with by optimizing the location CRLB. To do this, the cost function is defined by applying suitable operators (e.g., determinant, trace, or maximum eigenvalue) to the CRLB. When a single target is located inside the resolution cell, the optimization algorithm is simplified to the well-known Semi-definite Programming (SDP) using the related convex relaxation techniques.
\item An optimization algorithm for the multiple unresolved target case:\\
When multiple targets fall inside the same resolution cell, it is observed that the cost function is not convex anymore. In this case, due to the presence of sinusoidal terms in each entry of the Fisher-Information-Matrix (FIM), the cost function cannot be simplified into a convex form. Therefore, a sampling-based approach is proposed where initial conditions of the optimization algorithm are generated such that the algorithm moves towards the global minimum. Simulation results also confirm the efficacy of the proposed method in finding the optimum antenna allocation when multiple targets fall in the same or consecutive resolution cells.
\end{itemize}

The rest of this paper is organized as follows. Section II presents a brief overview of MIMO radars with collocated antennas. The CRLB is derived for the MIMO system in Section III. Section IV deals with the antenna allocation problem where the convex optimization framework for the single-target case is described. Simulation results are given in Section V. The paper is concluded in Section V.
\subsection{Notations}\label{Notation}
The notations used in this paper are as follows:
\begin{itemize}
\item $A=\mathcal{D}(\textbf{a})$: a diagonal matrix with $A_{ii}=\textbf{a}_i$ and $A_{ij}=0, \ i \neq j$
\item $\Re(a)$: the real part of the complex variable $a$
\item $\Im(a)$: the imaginary part of the complex variable $a$
\item $\mathcal{N}(\mu,\Sigma)$: a Gaussian function with mean $\mu$ and the covariance matrix $\Sigma$
\item $\mathcal{T}(A)$: the trace operator
\item $A^H$: the Hermition transpose
\item $A(:,i)$: the $i$-th column of matrix $A$
\end{itemize}
\section{MIMO Radars with collocated Antennas}
Consider an array of antennas with $M$ transmitters and $N$ receivers.
\begin{definition}\label{Location}
Define $\textbf{s}_{ti}=[x_{ti} \ y_{ti}]'$ and $\textbf{s}_{rj}=[x_{rj} \ y_{rj}]'$ as the location of the $i$-th transmitter and the $j$-th receiver in a 2-dimensional surveillance region, respectively.
\end{definition}
\begin{assumption}\label{Targets}
There are $T$ targets in the region where $\textbf{x}_t=[x_t \ y_t]'$ denotes the location of the $t$-th target.\footnote{Note that 3-D MIMO radars, although not very common in the literature, can be handled within our framework.} Also, the reflection of each target is modeled by a complex random variable $\alpha_t=\xi_t+j\zeta_t$ with $\xi$ and $\zeta$ being the real and imaginary parts of $\alpha$, respectively.
\end{assumption}
\begin{assumption}
It is assumed that the target's reflection follows a Swerling type I model \cite{Swerling1997} where $\{\xi_t\sim \mathcal{N}(\bar{\xi}_t,\sigma^2_{\alpha})\}$ and
$\{\zeta_t \sim \mathcal{N}(\bar{\zeta}_t,\sigma^2_{\alpha})\}$. Other models can be handled accordingly.
\end{assumption}
\begin{assumption}\label{CenterMass}
It is assumed that the distance between any two antennas is much smaller than the distance of the array to each target. It is also assumed that the arrays of transmitters and receivers are both collocated with the origin as the center of the mass of the array.
\end{assumption}
\begin{definition}
Define $\textbf{h}[k]=[h_1[k] \ \cdots \ h_M[k]]^H$ as the transmitted waveform in the $k$-th snapshot with $K$ being the number of total snapshots.
\end{definition}
\subsection{Signal Model}
Considering a collocated structure, resolution cells can be defined as a set of concentric circles where the radius of the $c$-th circle equals $c r_{bin}$ with $r_{bin}$ denoting the resolution width. Figure \ref{ResolutionCells} shows a simple configuration of resolution cells as well as the antennas that are distributed uniformly. Note that the target is located inside the $c$-th cell.
\begin{figure}
\centering
\includegraphics[width=8cm]{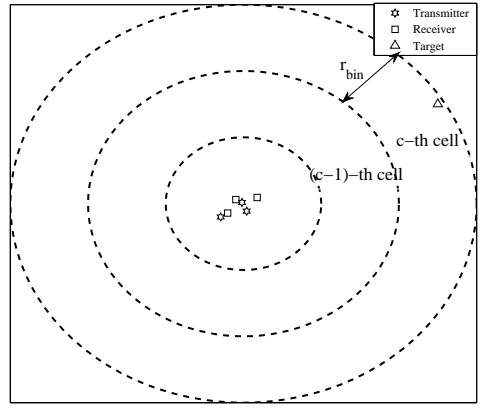}
\caption{A collocated MIMO radar with three transmitters and three receivers. The target is located inside the $c$-th cell. The resolution cells are shown as concentric circles with $cr_{bin}$ being the radius of the $c$-th cell.}
\label{ResolutionCells}
\end{figure}
\begin{assumption}
It is assumed that all $T$ targets are distributed in $C$ consecutive cells (e.g., $(c^*+1)$ to $(c^*+C)$) where $n_c$ denotes the number of targets inside the $c$-th cell. Without loss of generality, it is assumed that $c^*=0$.
\end{assumption}
\begin{assumption}
Transmitters send orthogonal signals with a diagonal cross-correlation matrix being defined as
\begin{equation}
R=\frac{1}{K}\sum_{k=1}^K \textbf{h}[k]\textbf{h}^H[k]=\mathcal{D}\left([P_1 \ \cdots \ P_M]'\right)
\end{equation}
where $P_m$ denotes the total transmitted power by the $m$-th antenna.
\end{assumption}
\begin{definition}
Defining $r^c_t=||\textbf{x}^c_t||_2$ as the Euclidean distance of the $t$-th target in the $c$-th cell to the origin, the ratio parameter $\beta^c_t$ is defined as follows:
\begin{equation}\label{Ratio}
\beta^c_t=\frac{r^c_t+(1-c)r_{bin}}{r_{bin}}
\end{equation}
\end{definition}
Now, given the above assumptions, the received output of the matched filter in the $c$-th resolution cell can be written as follows \cite{Gorji2011}:
\begin{equation}\label{SignalModel}
\eta_c=\left\{\begin{array}{cc}
\sum_{t=1}^{n_{c+1}}(1-\beta^{c+1}_{t})\phi^{c+1}_{t} & c=0\\
\sum_{t=1}^{n_{c}}\beta^{c}_{t}\phi^{c}_{t} & c=C\\
\sum_{t_1=1}^{n_c}\beta^c_{t_1}\phi^c_{t_1}+\sum_{t_2=1}^{n_{c+1}}(1-\beta^{c+1}_{t_2})\phi^{c+1}_{t_2} & \mbox{otherwise}
\end{array}
\right.+w
\end{equation}
where $w$ denotes a complex Gaussian noise with independent real and imaginary parts being distributed as $\{\Re(w),\Im(w)\} \sim \mathcal{N}(0,\sigma^2_w)$, and $\phi^c_t$ is the contribution of the $t$-th target in the signal received in the $c$-th cell, which is written as $\phi^c_t=\alpha^c_t \psi^c_t$ with the following form for the unknown term on the right-hand side of the equality \cite{Gorji2011}:
\begin{eqnarray}\label{TargetContribution}
\psi^c_t=\sqrt{K}\mbox{VEC}(A^c_t R^{\frac{1}{2}})
\end{eqnarray}
Here, $\mbox{VEC}(A)$ stands for the matrix vectorization operator, and $A^c_t$ denotes the steering matrix of the $t$-th target defined as follows \cite{Li2008}:
\begin{eqnarray}
A^c_t &=& \textbf{b}^c_t (\textbf{a}^c_t)^H\\
\textbf{a}^c_t &=& \exp\left(-j\frac{2\pi}{\lambda}[\sin(\theta^c_t) \ \cos(\theta^c_t)]S_{t}\right)\\
\textbf{b}^c_t &=& \exp\left(-j\frac{2\pi}{\lambda}[\sin(\theta^c_t) \ \cos(\theta^c_t)]S_{r}\right)
\end{eqnarray}
where $\lambda$ is the wavelength, $\theta^c_t$ denotes the DOA of the $t$-th target with respect to the origin, and the matrices $S_t$ and $S_r$ are defined as
\begin{eqnarray}
S_t &=& \left[\textbf{s}_{t1} \ \cdots \ \textbf{s}_{tM}\right]\\
S_r &=& \left[\textbf{s}_{r1} \ \cdots \ \textbf{s}_{rN}\right]
\end{eqnarray}
\begin{definition}\label{RealImagForm}
Given the vector of the output of the matched-filter as $\eta=[\eta^*_{1} \ \cdots \ \eta^*_{C}]^H$, define $\rho=[\Re(\eta_1) \ \Im(\eta_1) \ \cdots \ \Re(\eta_C) \ \Im(\eta_C)]'$.
\end{definition}
 Now, the mean received output of the matched filter is defined as $\bar{\rho}=[\Re(\bar{\eta}_1) \ \Im(\bar{\eta}_2) \ \cdots \ \Re(\bar{\eta}_C) \ \Im(\bar{\eta}_C)]'$.
The unknown terms $\Re(\bar{\eta}_c)$ and $\Im(\bar{\eta}_c)$ can be found by calculating $\Re(\bar{\phi}^c_t)$ and $\Im(\bar{\phi}^c_t)$ as follows and then replacing in (\ref{SignalModel}), respectively:
 \begin{eqnarray}\label{Average}
 \Re(\bar{\phi}^c_t) &=&  \bar{\xi}^{c}_t\Re(\psi^c_t)-\bar{\zeta}^{c}_t\Im(\psi^c_t)\nonumber\\
 \Im(\bar{\phi}^c_t) &=&  \bar{\xi}^{c}_t\Re(\psi^c_t)+\bar{\zeta}^{c}_t\Im(\psi^c_t)
 \end{eqnarray}
where the unknown terms on the right-hand side of the above equation can be written as follows:
\begin{eqnarray}\label{Psi}
\Re(\psi^c_t) &=& \sqrt{K}\cos\left(\frac{2\pi}{\lambda}[\sin(\theta^c_t) \ \cos(\theta^c_t)]\Omega(S_t,S_r,R)\right)\nonumber \\
\Im(\psi^c_t) &=& \sqrt{K}\sin\left(\frac{2\pi}{\lambda}[\sin(\theta^c_t) \ \cos(\theta^c_t)]\Omega(S_t,S_r,R)\right)
\end{eqnarray}
with $\Omega$ being defined as
\begin{equation}\label{Omega}
\Omega(S_t,S_r,R)=\left(1_{1\times M}\odot S_r-S_t \odot 1_{1\times N}\right)\left(R^{\frac{1}{2}} \odot 1_{1\times N}\right)
\end{equation}
where $\odot$ is the Kronecker product, and $1_{a\times b}$ stands for a $a\times b$ matrix with all entries being equal to one.

Given the signal model in (\ref{SignalModel}) and the mean output of the matched-filter in (\ref{Average}), the following proposition provides the distribution of the output of the matched-filter \cite{Gorji2011}:
\begin{proposition}
In a scenario with $T$ targets located in $C$ neighboring cells, the output of the matched-filter received by a collocated MIMO radar with $M$ transmitters and $N$ receivers (e.g., $\rho$) is Gaussian distributed with mean $\bar{\rho}$ and covariance $\Sigma$ defined as follows:
\begin{equation}\label{Covariance}
\Sigma=\left(\begin{array}{cccccc}
\Sigma_{11} & \Sigma_{12} & 0 & \cdots & 0\\
\Sigma_{21} & \Sigma_{22} & 0 & \cdots & 0\\
0 & 0 & \ddots & \cdots & 0\\
0 & 0 & \cdots & \Sigma_{(C-1)(C-1)} & \Sigma_{(C-1)C}\\
0 & 0 & \cdots & \Sigma_{C(C-1)} & \Sigma_{CC}
\end{array}
\right)
\end{equation}
with the following definitions for $\Sigma_{cc}$ and $\Sigma_{c(c-1)}$ terms:
\begin{equation}\label{PropForm}
\Sigma_{cc}=\left\{\begin{array}{cc}
K\sigma^2_{\alpha}\left(\sum_{t=1}^{n_1}(1-\beta^1_t)^2+\sigma^2_w\right)I_{2MN} & c=0 \nonumber \\
K\sigma^2_{\alpha}\left(\sum_{t=1}^{n_{C}}(\beta^C_t)^2+\sigma^2_w\right)I_{2MN} & c=C \nonumber \\
K\sigma^2_{\alpha}(\sum_{t_1=1}^{n_c}(\beta^c_{t_1})^2+ & \mbox{otherwise} \nonumber\\
\sum_{t_2=1}^{n_{c+1}}(1-\beta^c_{t_2})^2+\sigma^2_w)I_{2MN} & \nonumber\\
\end{array}
\right.
\end{equation}
\begin{equation}
\Sigma_{c(c-1)}=K\sigma^2_{\alpha}\sum_{t=1}^{n_c}(1-\beta^c_t)\beta^c_tI_{2MN}
\end{equation}
\end{proposition}
\section{Cramer-Rao Lower Bound}
The CRLB provides the best Minimum Mean Squared Error (MMSE) bound for any unbiased estimator \cite{BarShalom2001}. In this section, the CRLB is derived for a collocated MIMO radar. It is shown that the CRLB is a function of the distances between any two antennas. Also, a scenario is considered with $T$ targets distributed in $C$ consecutive cells where different number of targets might be located inside each cell.
\begin{definition}
For the $t$-th target located in the $c$-th resolution cell, define the state and parameter vector $X^c_t$ and $\Theta^c_t$, respectively, as follows:
\begin{eqnarray}
X^c_t &=& [x^c_t \ y^c_t \ \bar{\xi}^{c}_t \ \bar{\zeta}^{c}_t]' \\
\Theta^c_t &=& [\theta^c_t \ \beta^c_t \ \bar{\xi}^{c}_t \ \bar{\zeta}^{c}_t]'
\end{eqnarray}
\end{definition}
The CRLB is the inverse of the Fisher-Information-Matrix (FIM) defined as follows \cite{BarShalom2001}:
\begin{definition}
Assuming $\textbf{y}$ as the received noisy measurements and $\theta$ as the parameters of the measurement model, define the following matrix operator:
\begin{equation}\label{FIM}
J_{\theta\theta'} = E_{\textbf{y}}\left[\frac{\partial \log p(\textbf{y}|\theta)}{\partial \theta}\left(\frac{\partial \log p(\textbf{y}|\theta)}{\partial \theta}\right)'\right]
\end{equation}
\end{definition}
Refer to the definition of $\rho$ and its distribution provided by Proposition 1 and define the stacked state and parameter vector of all targets as $X=\left[(X^1_1)' \ \cdots \ (X^1_{n_1})' \ (X^2_{1})' \ \cdots \ (X^C_{n_C})'\right]'$ and $\Theta=\left[(\Theta^1_1)' \ \cdots \ (\Theta^1_{n_1})' \ (\Theta^2_{1})' \ \cdots \ (\Theta^C_{n_C})'\right]'$, respectively. In this case, the FIM is given by $J_{XX'}$. The defined FIM can be now written in the following form:
\begin{equation}
J_{XX'}=E_{\rho}\left[\frac{\partial \log p(\rho|X)}{\partial X}\left(\frac{\partial \log p(\rho|X)}{\partial X}\right)'\right]
\end{equation}
Using the chain-rule for partial derivatives, the above FIM can be simplified to the following form \cite{Gorji2011}:
\begin{equation}\label{FIM2}
J_{XX'}=\Gamma J_{\Theta\Theta'} \Gamma'
\end{equation}
Here, $\Gamma$ is called the system matrix and is written as
\begin{equation}\label{SystemMatrix}
\Gamma=\left[\begin{array}{cccc}
\gamma^1_1 & 0_{4\times4} & \cdots & 0_{4\times4}\\
0_{4\times4} & \gamma^1_2 & \cdots & 0_{4\times4}\\
\vdots & \vdots & \ddots & \vdots\\
0_{4\times4} & \cdots & 0_{4\times4} & \gamma^C_{n_C}
\end{array}
\right]
\end{equation}
 with $0_{4\times4}$ as a $4 \times 4$ zero-matrix, and individual $\gamma^c_t$ terms being defined as
\begin{equation}
\gamma^c_t=\left[\begin{array}{cccc}
\frac{\partial \theta^c_t}{\partial x^c_t} & \frac{\partial \beta^c_t}{\partial x^c_t} & 0 & 0\\
\frac{\partial \theta^c_t}{\partial y^c_t} & \frac{\partial \beta^c_t}{\partial y^c_t} & 0 & 0\\
0 & 0 & 1 & 0\\
0 & 0 & 0 & 1
\end{array}
\right]
\end{equation}
where the unknown partial derivatives can be derived using the definition of the ratio in (\ref{Ratio}) and the following equation for the target DOA:
\begin{equation}
\theta^c_t=\tan^{-1}\left(\frac{y^c_t}{x^c_t}\right)
\end{equation}
Now, the FIM derivation becomes finding the unknown term $J_{\Theta\Theta'}$ in (\ref{FIM2}). The new $J_{\Theta\Theta'}$ can be broken into the following sub-matrices:
\begin{equation}\label{NewForm}
J_{\Theta \Theta'}=\left[\begin{array}{cccccc}
J_{(\Theta^1)(\Theta^1)'} & J_{(\Theta^1)(\Theta^2)'} & 0 & 0 & \cdots & 0\\
J_{(\Theta^2)(\Theta^1)'} & J_{(\Theta^2)(\Theta^2)'} & J_{(\Theta^2)(\Theta^3)'} & 0 & \cdots & 0\\
0 & J_{(\Theta^3)(\Theta^2)'} & J_{(\Theta^3)(\Theta^3)'} & J_{(\Theta^3)(\Theta^4)'} & \cdots & 0\\
\vdots & \vdots & \vdots & \ddots & \ddots & \vdots\\
0 & 0 & \cdots & J_{(\Theta^{C-1})(\Theta^{C-2})'} & J_{(\Theta^{C-1})(\Theta^{C-1})'} & J_{(\Theta^{C-1})(\Theta^C)'}\\
0 & 0 & \cdots & 0 & J_{(\Theta^C)(\Theta^{C-1})'} & J_{(\Theta^C)(\Theta^C)'}
\end{array}
\right]
\end{equation}
Here, $\Theta^c$ denotes a $4 \times n_c$ vector formed by stacking the parameters of those targets falling inside the $c$-th cell. The following equation can be written for $\Theta^c$:
\begin{equation}
\Theta^c=[(\Theta^c_1)' \ \cdots \ (\Theta^c_{n_c})']'
\end{equation}
Each individual entry in (\ref{NewForm}) can be also written as follows:
\begin{eqnarray}\label{NewForm3}
J_{(\Theta^{c_1})(\Theta^{c_2})'}=\left[\begin{array}{ccc}
J_{(\Theta^{c_1}_1)(\Theta^{c_2}_1)'} & \cdots &  J_{(\Theta^{c_1}_1)(\Theta^{c_2}_{n_{c_2}})'}\\
\vdots & \ddots & \vdots \\
J_{(\Theta^{c_1}_{n_{c_1}})(\Theta^{c_2}_{n_{c_2}})'} & \cdots & J_{(\Theta^{c_1}_{n_{c_1}})(\Theta^{c_2}_{n_{c_2}})'}
\end{array}\right], \nonumber\\
c_1 \in \{c_2,c_2+1\}
\end{eqnarray}
Finally, each entry of the FIM in (\ref{NewForm3}) can be simplified into the following form:
\begin{equation}\label{Simplest}
J_{(\Theta^{c_1}_n)(\Theta^{c_2}_{m})'}=\left[\begin{array}{cccc}
J_{\theta^{c_1}_n \theta^{c_2}_m} & J_{\theta^{c_1}_n \beta^{c_2}_m} & J_{\theta^{c_1}_n \bar{\xi}^{c_2}_m} & J_{\theta^{c_1}_n \bar{\zeta}^{c_2}_m}\\
J_{\beta^{c_1}_n \theta^{c_2}_m} & J_{\beta^{c_1}_n \beta^{c_2}_m} & J_{\beta^{c_1}_n \bar{\xi}^{c_2}_m} & J_{\beta^{c_1}_n \bar{\zeta}^{c_2}_m}\\
J_{\bar{\xi}^{c_1}_n \theta^{c_2}_m} & J_{\bar{\xi}^{c_1}_n \beta^{c_2}_m} & J_{\bar{\xi}^{c_1}_n \bar{\xi}^{c_2}_m} & J_{\bar{\xi}^{c_1}_n \bar{\zeta}^{c_2}_m}\\
J_{\bar{\zeta}^{c_1}_n \theta^{c_2}_m} & J_{\bar{\zeta}^{c_1}_n\beta^{c_2}_m} & J_{\bar{\zeta}^{c_1}_n \bar{\xi}^{c_2}_m} & J_{\bar{\zeta}^{c_1}_n \bar{\xi}^{c_2}_m}\\
\end{array}
\right]
\end{equation}
Note that the matrix given by (\ref{Simplest}) is a $4 \times 4$ FIM sub-matrix that includes the information correlation between the parameters of the $n$-th target in the $c_1$-th cell and the $m$-th target in the $c_2$-th cell. Before presenting the algebraic expressions for each entry of the FIM given in (\ref{Simplest}), the following new notations are defined:
\begin{definition}\label{Notation1}
Assuming that the $n$-th target is located in the $c_1$-th cell, define the following new notations:
\begin{eqnarray}
\omega^{c_1}_n(l) &=& \frac{2\pi}{\lambda}\left[\sin(\theta^{c_1}_n) \ \cos(\theta^{c_1}_n)\right]\Omega(:,l)\\
\textbf{p}^{c_1}_n &=& \left[\cos(\theta^{c_1}_n) \ -\sin(\theta^{c_1}_n)\right]'\\
\breve{\beta}^{c_1}_n &=& \left[(1-\beta^{c_1}_n) \ \beta^{c_1}_n\right]'
\end{eqnarray}
where $\Omega(:,l)$ denotes the $l$-th column of matrix $\Omega$ with $\Omega(S_t,S_r,R)$ being written as $\Omega$ for brevity.
\end{definition}
\begin{definition}\label{Notation2}
For any two targets falling inside cells $c_1$ and $c_2$, respectively, the following notations are defined:
\begin{eqnarray}
\kappa^{nm}_{c_1c_2} &=& \bar{\xi}^{c_1}_n\bar{\xi}^{c_2}_m+\bar{\zeta}^{c_1}_n\bar{\zeta}^{c_2}_m \\
\iota^{nm}_{c_1c_2}&=& \bar{\xi}^{c_1}_n\bar{\zeta}^{c_2}_m-\bar{\zeta}^{c_1}_n\bar{\xi}^{c_2}_m
\end{eqnarray}
\end{definition}
The covariance matrix $\Sigma$ found in (\ref{Covariance}) can be now rewritten for cells $\{c_1-1,c_1,c_2\}$ with $c_1 \in \{c_2,c_2-1\}$. Using the general form given by (\ref{Covariance}) and expressions provided by (\ref{PropForm}), the new covariance matrix can be written as
\begin{equation}\label{sigmastar}
\Sigma_*=\left[\begin{array}{ccc}
c_1 & c_4 & 0\\
c_4 & c_2 & c_5\\
0 & c_5 & c_3
\end{array}
\right]\odot I_{2MN}
\end{equation}
where $c_i$ terms are found using (\ref{PropForm}). Similarly, the new notation $\bar{\rho}_*$ is defined as
\begin{equation}\label{rostar}
\bar{\rho}_*=[\Re(\bar{\eta}_{c_1-1}) \ \Im(\bar{\eta}_{c_1-1}) \ \Re(\bar{\eta}_{c_1}) \ \Im(\bar{\eta}_{c_1}) \ \Re(\bar{\eta}_{c_2}) \ \Im(\bar{\eta}_{c_2})]'
\end{equation}
Now, it can be shown that the inverse of $\Sigma_*$ can be written in the following form:
\begin{equation}\label{InverseForm}
\Sigma^{-1}_*=\left[\begin{array}{ccc}
k_1 & k_4 & k_5\\
k_4 & k_2 & k_6\\
k_5 & k_6 & k_3
\end{array}
\right]\odot I_{2MN}
\end{equation}
The following proposition provides algebraic expressions for each entry of the FIM in (\ref{Simplest}):
\begin{proposition}\label{CRLBProp}
Assume a scenario with $T$ targets falling inside $C$ consecutive resolution cells. Each entry of the
FIM defined by (\ref{Simplest}) can be calculated as follows:
\begin{eqnarray}\label{FIMs}
J_{\theta^{c_1}_n\theta^{c_2}_m} &=& K\left(\frac{2\pi}{\lambda}\right)^2\Bigg[\sum_{l=1}^{MN}(\textbf{p}^{c_1}_n)'\Omega(:,l)\Omega'(:,l)\textbf{p}^{c_2}_m  \times \nonumber\\ &\times& \left(\kappa^{nm}_{c_1c_2}\times\cos(\omega^{c_2}_m(l)-\omega^{c_1}_n(l))+\iota^{nm}_{c_1c_2}\sin(\omega^{c_2}_m(l)-\omega^{c_1}_n(l))\right)\Bigg]C_{\theta^{c_1}_n\theta^{c_2}_m}\\
J_{\beta^{c_1}_n\beta^{c_2}_m} &=& \frac{K}{r^2_{bin}}\left[\sum_{l=1}^{MN}\kappa^{nm}_{c_1c_2}\cos(\omega^{c_2}_m(l)-\omega^{c_1}_n(l))+\iota^{nm}_{c_1c_2}\sin(\omega^{c_2}_m(l)-\omega^{c_1}_n(l))\right]C_{\beta^{c_1}_n\beta^{c_2}_m}+F(\beta^{c_1}_n,\beta^{c_2}_m)\\
J_{\bar{\xi}^{c_1}_n\bar{\xi}^{c_2}_m} &=& J_{\bar{\zeta}^{c_1}_n\bar{\zeta}^{c_2}_m}=K\left[\sum_{l=1}^{MN}\cos\left(\omega^{c_1}_n(l)-\omega^{c_2}_m(l)\right)\right]C_{\theta^{c_1}_n\theta^{c_2}_m}\\
J_{\theta^{c_1}_n\beta^{c_2}_m} &=& \frac{K}{r_{bin}}\frac{2\pi}{\lambda}\left[\sum_{l=1}^{MN}(\textbf{p}^{c_1}_n)'\Omega(:,l)\left\{
\kappa^{nm}_{c_1c_2}\cos(\omega^{c_2}_m(l)-\omega^{c_1}_n(l))+\iota^{nm}_{c_1c_2}\sin(\omega^{c_2}_m(l)-\omega^{c_1}_n(l))
\right\}\right]C_{\theta^{c_1}_n\beta^{c_2}_m}\\
J_{\theta^{c_1}_n\bar{\xi}^{c_2}_m} &=& K\frac{2\pi}{\lambda}\left[\sum_{l=1}^{MN}(\textbf{p}^{c_1}_n)'\Omega(:,l)\left\{
-\bar{\zeta}^{c_1}_n\cos(\omega^{c_2}_m(l)-\omega^{c_1}_n(l))+\bar{\xi}^{c_1}_n\sin(\omega^{c_2}_m(l)-\omega^{c_1}_n(l))
\right\}\right]C_{\theta^{c_1}_n\theta^{c_2}_m}\\
J_{\theta^{c_1}_n\bar{\zeta}^{c_2}_m} &=& K\frac{2\pi}{\lambda}\left[\sum_{l=1}^{MN}(\textbf{p}^{c_1}_n)'\Omega(:,l)\left\{
\bar{\xi}^{c_1}_n\cos(\omega^{c_2}_m(l)-\omega^{c_1}_n(l))+\bar{\zeta}^{c_1}_n\sin(\omega^{c_2}_m(l)-\omega^{c_1}_n(l))
\right\}\right]C_{\theta^{c_1}_n\theta^{c_2}_m}\\
J_{\beta^{c_1}_n\bar{\xi}^{c_2}_m} &=& \frac{K}{r_{bin}}\left[\sum_{l=1}^{MN}\left\{
\bar{\xi}^{c_1}_n\cos(\omega^{c_2}_m(l)-\omega^{c_1}_n(l))+\bar{\zeta}^{c_1}_n\sin(\omega^{c_2}_m(l)-\omega^{c_1}_n(l))
\right\}\right]C_{\beta^{c_1}_n\theta^{c_2}_m}\\
J_{\beta^{c_1}_n\bar{\zeta}^{c_2}_m} &=& \frac{K}{r_{bin}}\left[\sum_{l=1}^{MN}\left\{
\bar{\zeta}^{c_1}_n\cos(\omega^{c_2}_m(l)-\omega^{c_1}_n(l))-\bar{\xi}^{c_1}_n\sin(\omega^{c_2}_m(l)-\omega^{c_1}_n(l))
\right\}\right]C_{\beta^{c_1}_n\theta^{c_2}_m}\\
J_{\bar{\xi}^{c_1}_n\bar{\zeta}^{c_2}_m} &=& K\left[\sum_{l=1}^{MN}\sin(\omega^{c_1}_n(l)-\omega^{c_2}_m(l))
\right]C_{\theta^{c_1}_n\theta^{c_2}_m}
\end{eqnarray}
with $F(\beta^{c_1}_n,\beta^{c_2}_m)$ being a known function of ratios, and the following expression being given for unknown coefficients in the right-hand side of the above equations:
\begin{eqnarray}
C_{\theta^{c_1}_n\theta^{c_2}_m} &=& \left\{\begin{array}{cc}
\left[(\breve{\beta}^{c_1}_n)' \ 0\right]\Sigma^{-1}_*\left[0 \ (\breve{\beta}^{c_2}_m)'\right]' & c_1=c_2-1\\
\left[0 \ (\breve{\beta}^{c_1}_n)'\right]\Sigma^{-1}_*\left[0 \ (\breve{\beta}^{c_2}_m)'\right]' & \mbox{otherwise}
\end{array}
\right.\\
C_{\beta^{c_1}_n\beta^{c_2}_m} &=& \left\{\begin{array}{cc}
\left[-1 \ 1 \ 0\right]\Sigma^{-1}_*\left[0 \ -1 \ 1\right]' & c_1=c_2-1\\
\left[0 \ -1 \ 1\right]\Sigma^{-1}_*\left[0 \ -1 \ 1\right]' & \mbox{otherwise}
\end{array}
\right.\\
C_{\theta^{c_1}_n\beta^{c_n}_m} &=& \frac{\partial C_{\theta^{c_1}_n\theta^{c_2}_m}}{\partial \beta^{c_2}_m}, \ C_{\beta^{c_1}_n\theta^{c_2}_m} = \frac{\partial C_{\theta^{c_1}_n\theta^{c_2}_m}}{\partial \beta^{c_1}_n}
\end{eqnarray}
\end{proposition}
\begin{proof}
See Appendix \ref{ProofCRLB}.
\end{proof}
Note that the above proposition can be used to find the FIM for every pair $\{\Theta^{c_1}_n,\Theta^{c_2}_m\}$ where $1 \leq \{c_1,c_2\} \leq C$ and $1 \leq n \leq n_{c_1}$ and $1 \leq m \leq n_{c_2}$. The FIMs calculated in (\ref{Simplest}) are then inserted in (\ref{NewForm3}) and (\ref{NewForm}), respectively, to obtain $J_{\Theta\Theta'}$. The CRLB is finally found by inverting the FIM as
\begin{equation}
C_{XX'}=(\Gamma^{-1})'C_{\Theta\Theta'}\Gamma^{-1}
\end{equation}
with $C_{\Theta\Theta'}=(J_{\Theta\Theta'})^{-1}$.
\section{Optimal Antenna Allocation}
\subsection{Motivation}
It can be shown that the localization performance of the collocated MIMO radar is affected by the distribution of antennas in the surveillance region. To illustrate this, consider a representative scenario with two antennas ($N=2, M=2$), where each antenna can both transmit and receive signals. We take a single target scenario into consideration with parameters $[30^o \ .33 \ 1 \ 1]'$, which is located in $\{r,\theta\}=[825\mbox{m} \ 30^o]'$. The variance of DOA estimates ($C_{\theta^2}$) is now shown in Figure \ref{CRLBSingleTest} in terms of different inter-antenna distances for the designed scenario. It can be observed that the geometry of sensors (inter-sensor distances) affects the performance bound of DOA estimation, where the estimation variance at the minimum point is $33\%$ lower than the maximum variance.

Unfortunately, the graphical tool cannot be developed for cases with more antennas. Therefore, this section concerns with designing a systematic algorithm for the antenna allocation problem in collocated MIMO radars. First, the case with a single target in the surveillance region is considered. It is shown that by considering suitable geometric constraints, the antenna allocation problem can be formulated as an SDP procedure \cite{Boyd2004}. Then, the problem is extended to the case with multiple targets in the same or consecutive resolution cells. It is shown that the derived cost function is non-convex and a sampling-based approach is proposed to capture the global minimum of the cost function.
\begin{figure}
\centering
\includegraphics[width=8cm]{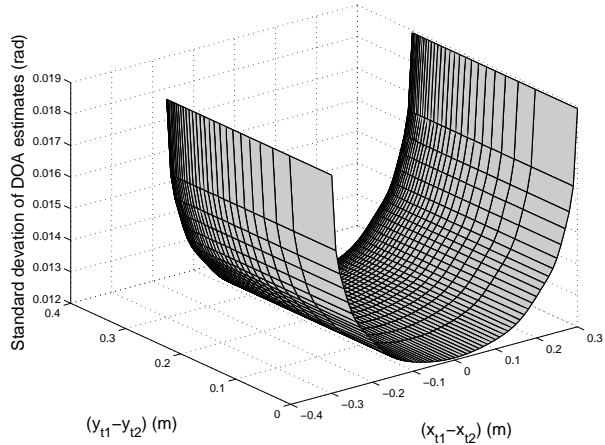}
\caption{The variance of the DOA estimation for different inter-sensor distances. The designed scenario includes a single-emitter and a collocated MIMO radar with two transmitters and two receivers.}
\label{CRLBSingleTest}
\end{figure}
\subsection{Single Target Case}
When a single target is placed in an arbitrary resolution cell, all $\cos(.)$ and $\sin(.)$ terms in the individual entries of the FIM defined by (\ref{Simplest}) vanish. Let us assume $\Theta=\left[\theta^{c} \ \beta^{c} \ \bar{\xi}^{c} \ \bar{\zeta}^c\right]'$ as the parameter vector of the single target fallen in the $c$-th cell. Using the results given in Proposition \ref{CRLBProp}, it can be observed that only terms $J_{\theta^c \nu^c}$ are a function of the antenna locations where $\nu^c \in \{\theta^c,\beta^c,\bar{\xi}^c,\bar{\zeta}^c\}$. On the other hand, according to the definition of the matrix $\Omega$ in (\ref{Omega}), one can show that:
\begin{equation}
\sum_{l=1}^{MN} \Omega(:,l)=0
\end{equation}
It can be observed that only the term $J_{(\theta^c)^2}$ can be considered as a function of the antenna locations.
\begin{definition}\label{InterSensorDefinition}
Define the difference between the $m$-th transmitter and the $n$-th receiver as $\Delta \textbf{s}_{nm}=\textbf{s}_{tm}-\textbf{s}_{rn}$.
\end{definition}
\begin{corollary}
In a collocated MIMO radar with $M$ transmitters and $N$ receivers, where a single target is located in the $c$-th resolution cell, the FIM is a function of inter-antenna differences.
In addition, all entries of the FIM are independent of the inter-sensor differences except $J_{(\theta^c)^2}$, which is also a convex function of the unknown differences.
\end{corollary}
\begin{proof}
It was shown that only $J_{(\theta^c)^2}$ is a function of the sensor locations. Now, it is demonstrated that it is a convex function of the parameters (difference vectors). Using the algebraic terms given by Proposition 2, the entry $J_{(\theta^c)^2}$ can be simplified into the following form:
\begin{equation}\label{SingleTargetForm1}
J_{(\theta^c)^2}=K\left(\frac{2\pi}{\lambda}\right)^2|\alpha^c|^2 \left[\sum_{l=1}^{nm}(\textbf{p}^c)'\Omega(:,l)\Omega'(:,l)\textbf{p}^c\right]C_{(\theta^c)^2}
\end{equation}
Consider the definition of $\Omega$ in (\ref{Omega}). It can be then observed that $\Omega(:,l)$ is a linear function of the corresponding difference vector $\Delta \textbf{s}_{nm}$. It is also known that $J_{(\theta^c)^2}$ is a convex function of $\Omega(:,l)$ terms due to the appearance of quadratic terms in (\ref{SingleTargetForm1}) \cite{Boyd2004}. Therefore, $J_{(\theta^c)^2}$ is also a convex function of the difference vectors.
\end{proof}
The antenna allocation problem can be now dealt with by minimizing the trace of CRLB, maximizing the determinant of FIM, or minimizing the maximum eigenvalue of CRLB \cite{Li2008}. The following lemma
proposes the convex optimization formulation for the antenna allocation problem in a collocated MIMO radar system where a single-target scenario is considered:
\begin{lemma}
Consider a collocated MIMO radar with $M$ transmitters and $N$ receivers. In addition, assume that there is a
single target located in the $c$-th resolution cell. Then, a convex optimization algorithm that finds an optimal placement of antennas is given as follows:
\begin{equation}\label{Optimization1}
\begin{array}{cc}
\max_{\left\{\Delta \textbf{s}_{11},\cdots,\Delta \textbf{s}_{nm}\right\}} & J_{(\theta^c)^2}
\end{array}
\end{equation}
\end{lemma}
\begin{proof}
The optimization problem can be formulated as minimizing the determinant of the CRLB, which is equivalent to maximizing $|J_{XX'}|$. In addition, the system matrix $\Gamma$ defined in (\ref{SystemMatrix}) is independent of the location of the antennas. Therefore, the final goal is to maximize $|J_{\Theta^{c}(\Theta^{c})'}|$. Now, the FIM in (\ref{Simplest}) can be written in the following new form:
\begin{equation}
J_{\Theta^{c}(\Theta^{c})'}=\left[\begin{array}{cc}
J_{(\theta^c)^2} & \textbf{b}'\\
\textbf{b} & B
\end{array}
\right]
\end{equation}
where $\textbf{b}$ and $B$ are block vector and matrix, respectively, formed by remaining entries of $J_{\Theta^{c}{\Theta^{c}}}$ in (\ref{Simplest}), respectively. The determinant term can be
written as
\begin{equation}
|J_{\Theta^{c}{\Theta^{c}}}|=|B||J_{(\theta^c)^2}-\textbf{b}'B\textbf{b}|
\end{equation}
It is known that both $B$ and $\textbf{b}$ are independent of the antenna placement. Therefore, the determinant maximization can be achieved by maximizing $J_{(\theta^c)^2}$ with respect to $\Omega$. However, it is also known that $\Omega$ is a linear function of $\Delta \textbf{s}_{nm}$ terms. The optimization problem can be finally simplified to
maximizing $J_{(\theta^c)^2}$ with respect to the $\Delta \textbf{s}_{nm}$ terms, which is the final form given in (\ref{Optimization1}).
\end{proof}
The final optimization problem can be now constructed by imposing the following constraints on the inter-antenna distances:

\textbf{The inter-antenna distance}:\\
In practice, antennas need to be well-separated to ensure maintenance and safety considerations. In addition, the inter-antenna distance should be small enough to have the far-field assumption still valid. Based on the given targets, the following constraints can be considered:
\begin{eqnarray}
||\Delta \textbf{s}_{nm}||_2 &\geq& d_{nm}  \\
||\Delta \textbf{s}_{nm}||_2 &\leq& e_{nm}, \forall \ m=1,...,M, \ n=1,...,N
\end{eqnarray}
where $e_{nm}$ and $d_{nm}$ are design parameters.

\textbf{The center of the mass constraint}:\\
It was mentioned in Assumption \ref{CenterMass} that the center of the mass of the array is located in the origin. Therefore, the following new constraints are formed on the location of antennas:
\begin{eqnarray}\label{CenterMass2}
\sum_{m=1}^M \textbf{s}_{tm}+\sum_{n=1}^N \textbf{s}_{rn}=0
\end{eqnarray}
Note that the FIM is a function of inter-antenna distances and therefore, a set of optimal difference vectors might correspond to an infinite number of sensor locations. The constraint given by (\ref{CenterMass2}) ensures that the mass center of the obtained geometry is in the origin. The uniqueness of optimal solution is further discussed in this section.

Considering the above defined constraints, the new optimization problem can be written as follows:
\begin{equation}\label{Optimization3}
\begin{array}{cc}
\max_{\left\{\Delta \textbf{s}_{11},\cdots,\Delta \textbf{s}_{nm}\right\}} & \sum_{m=1}^M\sum_{n=1}^N (\textbf{p}^c)'\Delta \textbf{s}_{nm} \Delta \textbf{s}'_{nm}\textbf{p}^c\\
\mbox{S.T} & ||\Delta \textbf{s}_{nm}||_2 \geq d_{nm}\\
 & ||\Delta \textbf{s}_{nm}||_2 \leq e_{nm}\\
 & \sum_{m=1}^M \textbf{s}_{tm}+\sum_{n=1}^N \textbf{s}_{rn}=0, \forall \ m=\{1,\cdots,M\}, n=\{1,\cdots,N\}
\end{array}
\end{equation}
In writing the above equation, it is assumed that the transmitted powers are all the same and unitary ($P_1=P_2=\cdots=P_M=1$).
The optimization problem given by (\ref{Optimization3}) is not convex and therefore cannot be solved using the standard approaches. The following theorem reformulates the optimization problem in (\ref{Optimization3}) as an SDP:
\begin{thm}\label{Theorem1}
Consider a single-target scenario with a collocated MIMO radar being used as the measurement tool. Defining $T^*=\{T_{11},\cdots,T_{nm}\}$,
$S^*=\{\textbf{s}_{t1},\cdots,\textbf{s}_{rN}\}$, and $\textbf{t}=\left[t_{11} \ \cdots \ t_{nm}\right]'$,
the optimal placement of transmitters and receivers that maximizes the determinant of FIM is found by solving the
following SDP optimization problem:
\begin{equation}
\begin{array}{cc}
\max_{T^*,S^*,\textbf{t}} & \sum_{m=1}^M \sum_{n=1}^N t_{nm}\\
\mbox{S.T.} & \sum_{m=1}^M \textbf{s}_{tm}+\sum_{n=1}^N \textbf{s}_{rn}=0\\
 & \mathcal{T}\left(T_{nm}P\right) \geq t_{nm}\\
 & \left[\begin{array}{cc}
-I_{2\times 2} & \textbf{s}_{tm}-\textbf{s}_{rn}\\
\textbf{s}'_{tm}-\textbf{s}'_{rn} & -e^2_{nm}
\end{array}\right]\preceq 0, \ \left[\begin{array}{cc}
I_{2\times 2} & \textbf{s}_{tm}-\textbf{s}_{rn}\\
\textbf{s}'_{tm}-\textbf{s}'_{rn} & d^2_{nm}
\end{array}\right]\preceq 0\\
& \nonumber \\
& \left[\begin{array}{cc}
1 & \textbf{s}'_{tm}-\textbf{s}'_{rn}\\
\textbf{s}_{tm}-\textbf{s}_{rn} & T_{nm}
\end{array}\right]\preceq 0, \ \forall \ m=\{1,\cdots,M\}, n=\{1,\cdots,N\}\\
\end{array}
\end{equation}
with $P=\textbf{p}^c(\textbf{p}^c)'$, and $\preceq$ as the generalized inequality operator.
\end{thm}
\begin{proof}
See Appendix \ref{ProofTheorem1}.
\end{proof}
The above optimization problem can be now efficiently solved using standard packages for solving SDP problems \cite{Grant2002}.
\begin{remark}
The optimization problem in (\ref{Optimization3}) proves the dependency on the parameters of the target through the matrix $P$. The following proposition shows how the optimal structure is affected by changing the DOA of the target:
\begin{proposition}\label{RotationLemma}
Consider a single-target scenario with a collocated MIMO radar being used as the measurement tool. Defining $\theta_1$ and $\theta_2$ as two different DOAs and $\{S^{o1}_t,S^{o1}_r\}, \{S^{o2}_t,S^{o2}_r\}$ as the assigned optimal antenna allocations, respectively, the following equations are valid:
\begin{eqnarray}
\textbf{s}^{o2}_{tm} &=& G_{\Delta \theta}  \textbf{s}^{o1}_{tm} \\
 \textbf{s}^{o2}_{rn} &=& G_{\Delta \theta}  \textbf{s}^{o1}_{rn}
 \ \forall \ m=\{1,\cdots,M\}, n=\{1,\cdots,N\}
\end{eqnarray}
with $\Delta \theta=\theta_2-\theta_1$ and $G_{\Delta \theta}$ as the rotation matrix defined as follows:
\begin{equation}
G_{\Delta \theta} = \left[\begin{array}{cc}
\cos(\Delta \theta) & -\sin(\Delta \theta)\\
\sin(\Delta \theta) & \cos(\Delta \theta)
\end{array}
\right]
\end{equation}
\end{proposition}
\begin{proof}
See Appendix \ref{Rotation}.
\end{proof}
\end{remark}
\begin{remark}
The SDP formulation given by Theorem 1 does not provide any information regarding the uniqueness of the optimal solutions for the location of antennas. The uniqueness of solutions is now discussed in the following proposition.
\begin{proposition}\label{Uniqueness}
Consider a single-target scenario with a collocated MIMO radar being used as the measurement tool. Then, there are at least two solutions for the optimization problem in (\ref{Optimization3}) as $\{S_t^{o1},S_r^{o1}\}$ and $\{S_t^{o2},S_r^{o2}\}$. In addition, the first optimal configuration can be obtained from the second one by a simple rotation as follows:
\begin{eqnarray}
\textbf{s}^{o2}_{ti} &=& G_{\pi}\textbf{s}^{o2}_{ti}, \ i=\{1,\cdots,M\} \\
\textbf{s}^{o2}_{rj} &=& G_{\pi}\textbf{s}^{o2}_{rj}, \ j=\{1,\cdots,N\}
\end{eqnarray}
where $G_{\pi}$ is a rotation matrix with $\pi$ as the angle of rotation.
\end{proposition}
\begin{proof}
See Appendix \ref{UniquenessProof}.
\end{proof}
\end{remark}
\subsection{Multiple-target Case}
When multiple targets fall inside the same resolution cell (or consecutive cells), the individual entries of the FIM in (\ref{Simplest}) are no longer convex.
\begin{proposition}
Consider a collocated MIMO radar system with $M$ transmitters and $N$ receivers with $d_{nm} \leq ||\Delta \textbf{s}_{nm}|| \leq e_{nm}, \ \forall \ m=\{1,\cdots,M\}, n=\{1,\cdots,N\}$. Also, assume a scenario with two targets in the $c$-th resolution cell with parameters $\Theta^c_1$ and $\Theta^c_2$, respectively. Then, the term $(\omega^{c}_1(:,l)-\omega^{c}_2(:,l))$ falls in the following interval:
\begin{equation}\label{Bound}
\frac{2\pi}{\lambda}d_{nm}\sqrt{2\left(1-\cos(\theta^c_2-\theta^c_1)\right)} \leq \omega^{c}_1(:,l)-\omega^{c}_2(:,l) \leq
\frac{2\pi}{\lambda}e_{nm}\sqrt{2\left(1-\cos(\theta^c_2-\theta^c_1)\right)}
\end{equation}
with $l=\{1,\cdots,MN\}$.
\end{proposition}
The above proposition states that the more separated the DOA of targets, the wider $(\omega^{c}_1(:,l)-\omega^{c}_2(:,l))$. For example, defining $\Delta \omega$ as the difference between the upper and lower bounds of $(\omega^{c}_1(:,l)-\omega^{c}_2(:,l))$ in (\ref{Bound}), Figure \ref{DeltaOmega} shows how $\Delta \omega$ changes by varying the difference between the DOA of targets. It is observed that when the targets are well-separated in the DOA space, the difference between the maximum and minimum bound is significant. This also highlights the contribution of the sinusoid terms in each entry of the FIM, which might result in several local optimum points. On the other hand, the convex relaxation approach used for the single target case cannot be applied to the cost function derived for the case with multiple targets in the same cell. The above problems make the optimization problem non-convex when there are more than one target inside each resolution cell (consecutive cells).

To handle the above problem, the optimization algorithm is solved for different initial locations of the antennas. However, a large number of initial points are required to capture the non-convexity of the cost function. The sampling approach is now proposed in Algorithm \ref{Algorithm1}. In the proposed algorithm, $Q$ denotes the covariance of the normal density function that is used to generate new initial points. While the covariance matrix is chosen experimentally, a small variance might make the algorithm be trapped in the local optimum point. Therefore, an intelligent choice of the covariance matrix can enhance the efficiency of the algorithm. The main idea behind the proposed approach is to, first, find an estimate of the optimal antenna location, which might be a local solution. Then, initial points are generated based on the obtained optimal location.
\begin{remark}
Note that the proposed algorithm does not always guarantee that the optimization algorithm captures the global solution. Due to the non-convexity of the cost function, there is also no analytical way to capture the global solution. Nevertheless, the proposed sampling approach initializes new points around the initial local solutions and pushes the overall algorithm into the global solution. As shown in the simulations, it can be observed that regardless of the initial selection of the antenna location, the algorithm always converges to a unique solution.
\end{remark}
\begin{remark}
The procedure given in Algorithm \ref{Algorithm1} is terminated when the cost function is not reduced in more than $\mu$ iterations where the parameter $\mu$ is empirically chosen. If the algorithm finds the global solution of the cost function, randomly-generated initial conditions around the optimal point does not give a lower cost and, therefore, the algorithm does not advance in subsequent iterations. In this case, the procedure is stopped after $\mu$ unsuccessful trials.
\end{remark}
\begin{figure}
\centering
\includegraphics[width=8cm]{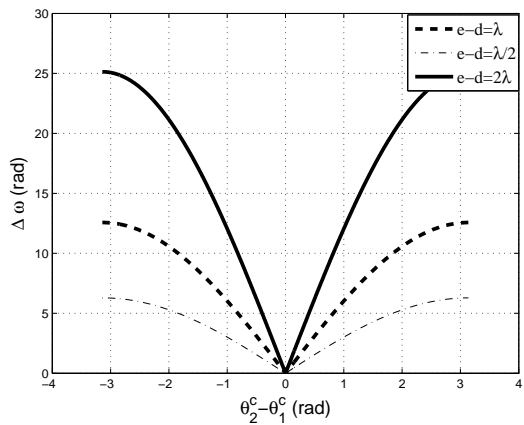}
\caption{Variation of the sinusoid argument by changing the difference between the DOA of targets.}
\label{DeltaOmega}
\end{figure}

\begin{algorithm}
\caption{The optimization algorithm for the case with multiple targets in the same resolution cell} \label{Algorithm1}
\begin{algorithmic}
\STATE \textbf{Initialization}: Generate an initial location of antennas as $\textbf{s}_{tm}^0,\textbf{s}_{rn}^0$ with $m=\{1,\cdots,M\}$ and $n=\{1,\cdots,N\}$.
\STATE \textbf{Optimization}: Find an optimal distribution of antennas by minimizing the following cost function:
\begin{equation}
\begin{array}{cc}
\min_{\textbf{s}_{t1},\cdots,\textbf{s}_{rN}} & \mathcal{T}(C_{XX'})\\
S.T. & ||\Delta \textbf{s}_{nm}|| \geq d_{nm}\\
 &  ||\Delta \textbf{s}_{nm}|| \leq e_{nm}\\
 & \sum_{m=1}^M \textbf{s}_{tm}+\sum_{n=1}^N \textbf{s}_{rn}=0, \ \forall \ m=\{1,\cdots,M\}, \ n=\{1,\cdots,N\}
\end{array}
\end{equation}
\STATE \textbf{Optimal Cost}: Initialize $\textbf{s}_{tm}^o$ and $\textbf{s}_{rn}^o$ and calculate the assigned cost as $\mathcal{C}^o=\mathcal{T}(C_{XX'})_{\textbf{s}_{t1}^o,\cdots,\textbf{s}_{rN}^o}$.
\STATE \textbf{Sampling}: While $u \leq U$ or $NA<\mu$:
\begin{itemize}
\item Sample $\textbf{s}^{0}_{tm} \sim \mathcal{N}(\textbf{s}^o_{tm},Q)$ and $\textbf{s}^{0}_{rn} \sim \mathcal{N}(\textbf{s}^o_{rn},Q)$ with $m=\{1,\cdots,M\}$ and $n=\{1,\cdots,N\}$.
\item Run the optimization algorithm and find the new distribution of antennas $\textbf{s}_{tm}^{\star},\textbf{s}_{rn}^{\star}$ and associated cost $\mathcal{C}^{\star}$.
\IF {$\frac{\mathcal{C}^{\star}}{\mathcal{C}^0} \leq 1$}
\STATE $\textbf{s}_{tm}^o=\textbf{s}_{tm}^{\star}$, $\textbf{s}_{rn}^o=\textbf{s}_{rn}^{\star}$ with $m=\{1,\cdots,M\}$ and $n=\{1,\cdots,N\}$.
\STATE $\mathcal{C}^o=\mathcal{C}^{\star}$.
\STATE $NA=0$.
\ELSE
    \STATE $NA=NA+1$.
\ENDIF
\end{itemize}
\STATE Report $\textbf{s}^o_{tm}$ and $\textbf{s}^o_{rn}$ as the optimal distribution of antennas.
\end{algorithmic}
\end{algorithm}
\section{Simulation Results}
In this section, we analyze how the optimal allocation of antennas in the surveillance region affects the localization performance of the MIMO radar system. To do this, a collocated MIMO radar is first designed with the parameters in Table \ref{TableParmaeters}. In the following subsections, the performance of the optimization algorithm is studied first for a single target scenario. Then, the simulations results will be provided for a scenario with multiple targets occupying the same resolution cell.
\begin{table}
\centering \caption{Simulation Parameters}\label{TableParmaeters}
\begin{tabular}{|c|c|c|}
\hline \hline
\textbf{Parameter} & \textbf{Description} & \textbf{Value}\\
\hline
$r_{max}$ & Maximum coverage range of transmitters & $5$ (km)\\
\hline
$r_{bin}$ & Range width& $30$ (m)\\
\hline
$\lambda$ & Wave-length& $30$ (cm)\\
\hline
$K$ & Number of snapshots & $128$\\
\hline
$\sigma^2_{\alpha}$ & Variance of the scatterers & $10^{-4}$\\
\hline
$\sigma^2_w$ & Variance of the additive noise & $1$\\
\hline
$P_m$ & Transmitted power & $1$ (W)\\
\hline \hline
\end{tabular}
\end{table}
\subsection{A Single-Target Scenario}
Initially, consider a single target located at $[410 \ -710]'$ (m). The parameters of the target are also chosen as follows:
\begin{equation}
\Theta=[-\frac{\pi}{3} \ .33 \ 3 \ 3]
\end{equation}
In the first experiment, assume that there are $M$ antennas where each antenna can both transmit and receive signals. Two antenna configurations, a Uniform-Linear-Array (ULA) with half wave-length spacing and the optimal geometry proposed in this paper, are considered in this part. For simulations, it is assumed that
$d_{mn}=\lambda, e_{mn}=2\lambda\ \forall \{m,n\}$. The optimal configuration of antennas is shown in
Figure \ref{OptimalConfgSingleAntenna} for different number of antennas. In addition, Figure \ref{CRLBOptimalVsULA} presents the location CRLB for both optimal and ULA structures separately. It can be observed that the CRLB of the optimal
configuration is much lower than that of the ULA structure. The improvement becomes more significant when the number of antennas is smaller. For example, for the case with $M=2$ antennas, the CRLB of the optimal
structure is around $6$ times lower than that of the ULA configuration while the improvement decays to $2$ times lower at $M=5$ antennas. When the number of antennas increases, the gap between the optimal and ULA CRLB becomes smaller because the Signal-to-Noise Ratio (SNR) is large enough to
make up the poor geometry of antennas.
\begin{figure}
\centering
\includegraphics[width=10cm]{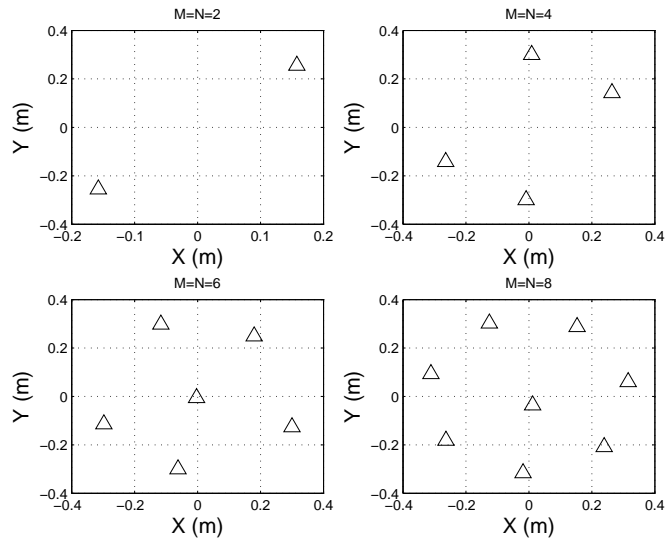}
\caption{The optimal configuration of antennas for a single-target case. The optimal configuration is found for different number of antennas where each antenna can both transmit and receive signals.}
\label{OptimalConfgSingleAntenna}
\end{figure}
\begin{figure}
\centering
\includegraphics[width=8cm]{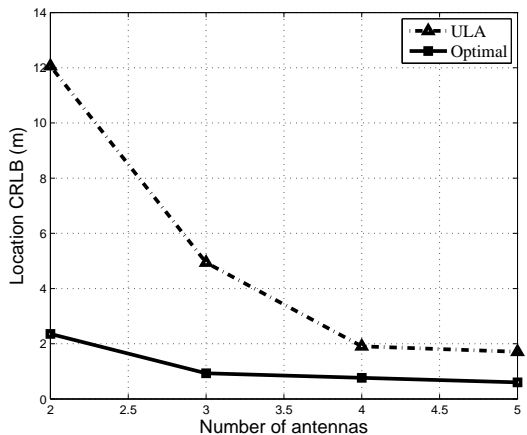}
\caption{Localization CRLB for the ULA configuration and the optimal structure. The CRLB is found for a single-target scenario and different number of antennas.}
\label{CRLBOptimalVsULA}
\end{figure}
\subsubsection{The effect of the target DOA on the optimal structure}
 Consider the above scenario with $M=4$ antennas. While the target is still assumed to fall in the same cell defined in the above experiment, its DOA varies in the interval $[\frac{-\pi}{2},\frac{\pi}{2}]$. The optimization algorithm is now
implemented to find the optimal configuration of antennas. Figure \ref{OptimalConfigDifferentDOAsSingleTarget} shows the results for four different target DOAs. The results shown in Figure \ref{OptimalConfigDifferentDOAsSingleTarget} imply that the optimal configuration with $\theta_1$ as the DOA can be obtained from the optimal structure with $\theta_2$ by rotating the geometry $(\theta_2-\theta_1)$ (rad) around the mass center, which confirms Proposition \ref{RotationLemma}.
\subsubsection{The localization performance of the optimal structure}
Assume $M=3$ for the number of antennas.
Besides the optimal and ULA configurations, a random antenna allocation is also used for the test where the antennas are randomly distributed in the underlying surveillance region. The localization Root Mean Squared Error (RMSE) is now calculated at different target
SNRs where all results are obtained after $100$ Monte Carlo simulations. Figure \ref{RMSEOptimalConfig} presents the resulting RMSE for each of the above configurations. It is observed that the optimal configuration achieves the lowest
RMSE while the ULA provides the worst results. The random allocation also gives an RMSE between the optimal and ULA configurations although other random distributions of antennas may provide higher RMSE results.
\subsubsection{The optimal design for separate transmitter and receiver arrays}
Simulation results are now provided for a scenario in which each antenna can either transmit or receive signals. Consider a single-target scenario with $\theta^c=-\frac{\pi}{3}$ (rad) as the DOA. The optimal structure is now found for two cases with $M=N=2$ and $M=N=6$ antennas. Figure \ref{OptimalConfgSingleAntennaSeparate} presents the obtained optimal structures where, for each case, the results are given for scenarios with the same and separate transmitters and receivers, respectively. It can be observed that the optimal structure obtained for each case (e.g., the same and separate transmitters and receivers) is the same with transmitters and receivers being clustered such that the mutual distances between the same-type antennas (e.g., transmitter or receiver) is minimized. To test this hypothesis, assume that $6$ antennas are available and there are two scenarios with $M=4$ and $M=3$ as the number of transmitters at each scenario. The optimal structure is now found for each scenario and the final results are shown in Figure \ref{OptimalConfgSingleAntennaSeparate2}. It can be observed that the same optimal structure is obtained for both cases with antennas being clustered based on the mutual distances between the antennas with the same type.

Note that although the obtained optimal structures in Figure \ref{OptimalConfgSingleAntennaSeparate2} are similar, the optimum cost function might be different based on the number of signal paths ($M\times N$). For example, for the configurations given in Figure \ref{OptimalConfgSingleAntennaSeparate2}, the optimal cost is calculated to be $0.7545$ and $0.6393$ for $M=4$ ($M\times N=8$) and $M=3$ ($M \times N = 9$) antennas, respectively. The obtained optimum cost values also confirm the fact that the more the diversity gain, the lower the achieved optimum cost.
\begin{figure}
\centering
\includegraphics[width=10cm]{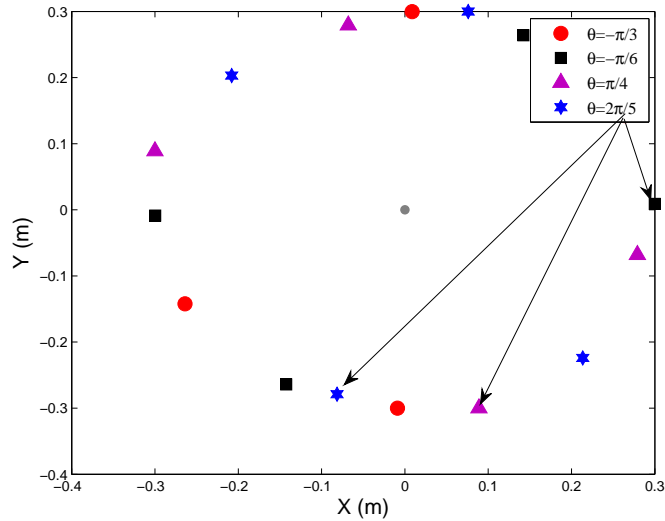}
\caption{The optimal configuration of antennas for a single-target case, and for four different target DOAs.}
\label{OptimalConfigDifferentDOAsSingleTarget}
\end{figure}
\begin{figure}
\centering
\includegraphics[width=10cm]{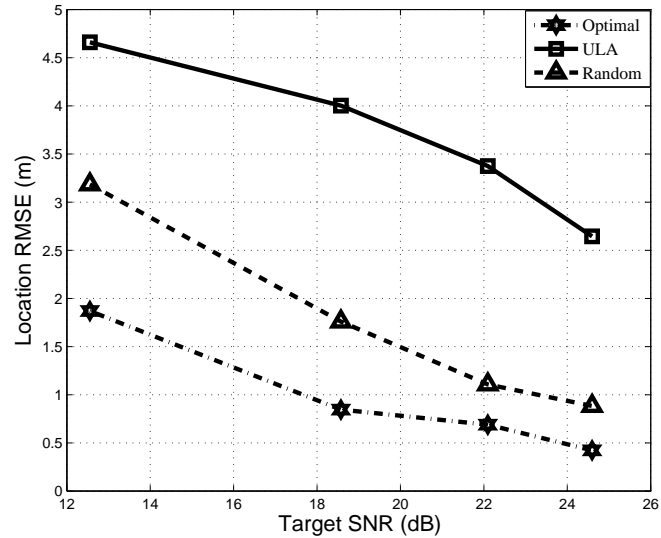}
\caption{Location RMSE for different target SNRs and for the MIMO radar with $M=3$ antennas. The RMSE results are obtained for three different structures (ULA, optimal, and randomly-distributed configurations).}
\label{RMSEOptimalConfig}
\end{figure}
\begin{figure}
\centering
\includegraphics[width=10cm]{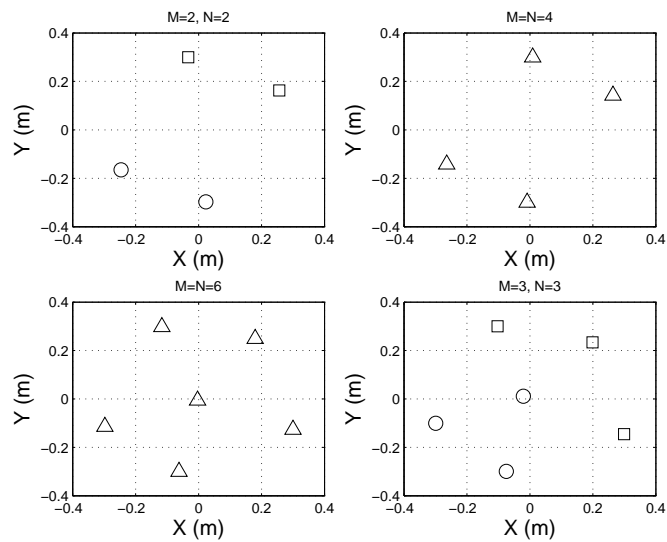}
\caption{The optimal configuration of antennas for a single-target case. The optimal configuration is found for different number of antennas where each antenna can either transmit or receive signals ($\bigcirc$- transmitters, $\Box$- receivers, and $\bigtriangleup$- transceivers).}
\label{OptimalConfgSingleAntennaSeparate}
\end{figure}

\begin{figure}
\centering
\includegraphics[width=10cm]{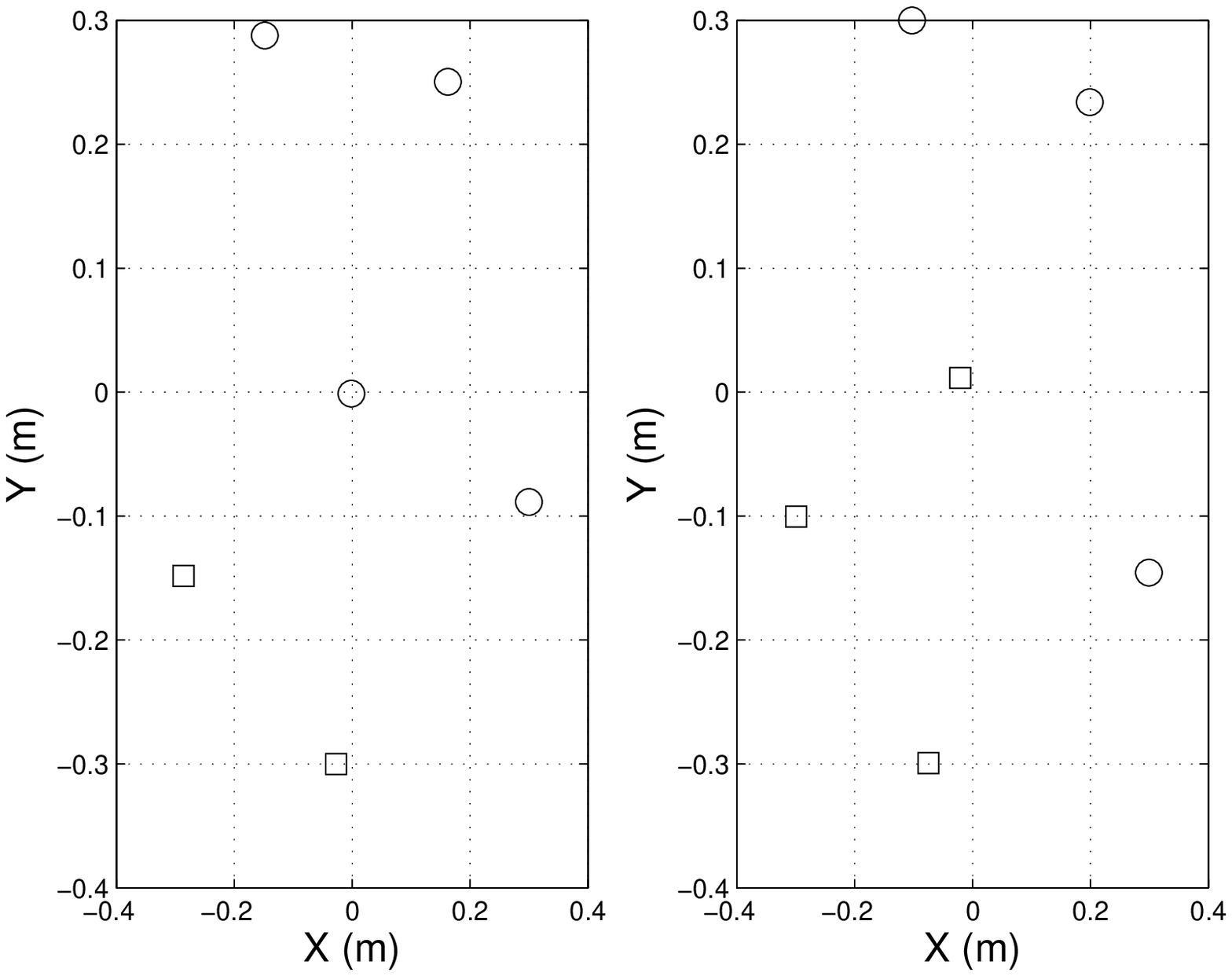}
\caption{The optimal antenna configuration for the single-target scenario with $M+N=6$ antennas. The optimal structure is found for two cases with $M=4$ and $M=3$ antennas as the number of transmitters ($\bigcirc$- transmitters and $\Box$- receivers).}
\label{OptimalConfgSingleAntennaSeparate2}
\end{figure}
\subsection{Multiple-Target Case}
In this subsection, the optimization algorithm is applied to a scenario with more than one target being located in the same resolution cell. Let us assume there are two targets falling in the same cell with the following parameters:
\begin{eqnarray}
\Theta^c_1 &=& \left[-\frac{\pi}{3} \ .33 \ 3 \ 3\right]' \nonumber\\
\Theta^c_2 &=& \left[+\frac{\pi}{3} \ .66 \ 3 \ 3\right]'
\end{eqnarray}

Based on the results in Figure \ref{DeltaOmega}, it is now evident that the effect of sinusoidal terms on the cost function cannot be ignored due to the large value for $\Delta \theta$. First, the optimization framework given by Algorithm \ref{Algorithm1} is applied to the two-target scenario with different initial conditions. Figure \ref{InitialCondition} shows the cost values obtained at different iterations of the algorithm and for different initial conditions. It is observed that the algorithm captures the global minimum after a number of iterations. While each initial condition leads to a different cost value, the sampling approach finally finds the structure corresponding to the global minimum. Note that without the sampling procedure, each initial condition leads to a different optimal cost as shown in Figure \ref{OptimalCostDifferentInitial}.
\begin{figure}
\centering
\includegraphics[width=10cm]{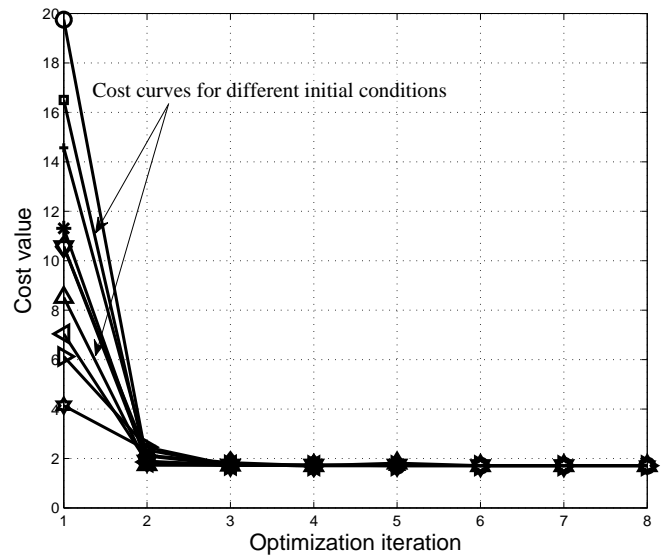}
\caption{The cost for $10$ different initial antenna locations (different symbols correspond to initial conditions). The simulations are done for a two-target scenario the same resolution cell.}
\label{InitialCondition}
\end{figure}
\begin{figure}
\centering
\includegraphics[width=10cm]{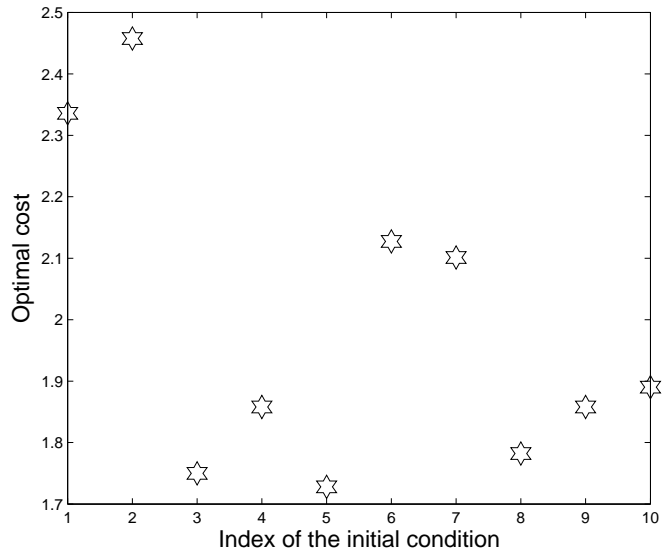}
\caption{The optimal cost for $10$ different initial antenna locations without using the sampling-based approach.}
\label{OptimalCostDifferentInitial}
\end{figure}

\subsubsection{The effect of the angular separation on the optimal configuration}
The optimal structure is found for different values of $\Delta \theta=\theta^c_2-\theta^c_1$. The optimal configurations are now depicted in Figure \ref{OptimalStructureDeltaTheta} for four different values of $\Delta \theta$. It is observed that when $\Delta \theta \rightarrow 0$, the obtained structure resembles the one given in Figure \ref{OptimalConfgSingleAntenna} for the scenario with $M=4$ antennas. Nevertheless, for other values of $\Delta \theta$, a new structure whose geometry depends on the distribution of targets in the resolution cell is obtained. Figure \ref{CRLBDeltaTheta} also presents the cost function (i.e., the trace of the location CRLB) for different values of $\Delta \theta$ where the results are obtained for both optimal and ULA configurations. The graph indicates that the closer the targets the poorer the performance. For example, the cost at $\Delta \theta=\frac{\pi}{100}$ is $100 \%$ higher than the one at $\Delta \theta=\frac{\pi}{50}$ when the optimal structure is taken. In addition, Figure \ref{CRLBDeltaTheta} confirms the superiority of the optimal structure to the ULA. The rate of the improvement also increases when targets become more closer, with $10$ times lower cost at $\Delta \theta=\frac{\pi}{100}$ compared to $5$ times lower cost achieved at $\Delta \theta=2\frac{\pi}{3}$.

The obtained results in Figures \ref{OptimalConfgSingleAntenna} and \ref{CRLBDeltaTheta} imply that although the optimization algorithm can be implemented more efficiently when the angular separation between two targets becomes smaller, the performance of the localization is degraded. In other words, there is a trade-off between the quality of the localization and the efficiency of the optimization algorithm. Smaller values of $\Delta \theta$ makes the FIM entries in (\ref{Simplest}) less dependent on the sinusoid terms.

\subsubsection{The optimal design for multiple unresolved targets}
The performance of the optimization algorithm can be also evaluated for a scenario with more than two targets inside the same resolution cell. It is known that there is a bound on the maximum number of targets that can be uniquely detected in the same resolution cell \cite{Li2007V2}. Assume different number of targets are placed in the $c$-th resolution cell with the same SNR being assigned to each target. Also, consider the MIMO structure with $M=N=4$ antennas where each antenna can both transmit and receive signals. We find the optimal structure for each case with a different number of targets inside the same resolution cell.

For comparison, the localization algorithm is also applied to the obtained structures and the location RMSE is calculated by taking an average of individual estimates in $100$ Monte Carlo runs. The RMSE results as well as the location CRLB are now depicted in Figure \ref{RMSECRLBULAOptimal} where the graphs for the case with the ULA MIMO structure are also included. While the localization performance degrades, when the number of targets increases, the optimal structure always shows the lower RMSE compared to the ULA configuration. In addition, when more targets fall inside the same resolution cell, the difference between the obtained RMSE of the ULA structure and that of the optimal configuration becomes higher. For example, for the scenario with $T=2$ unresolved targets, the optimal RMSE is $53 \%$ lower than the RMSE obtained by the ULA structure. Nevertheless, the gap widens to $123 \%$ when $5$ targets occupy the same resolution cell. Although the distribution of targets in the cell also affects the localization performance \cite{Gorji2011}, this experiment shows the superiority of the optimal structure compared to the ULA configuration, specially, when more targets are placed in the same resolution cell.
\section{Conclusions}
This paper considered the antenna allocation problem in a collocated MIMO radar system. A novel derivation of the CRLB was presented where both range and DOA information were included in the CRLB. An SDP problem was then formulated for antenna allocation when a single target is located inside the resolution cell. Then, the antenna allocation was extended to the multiple unresolved target scenarios, and it was shown that the final cost function is non-convex. A sampling-based approach was proposed to capture the global minimum of the proposed cost function. Simulation results were also presented for both scenarios with the single-target and multiple targets occupying the same resolution cell. The obtained results confirmed the superiority of the optimal configuration compared to the common ULA structure in both single and multiple target scenarios.
\begin{figure}
\centering
\includegraphics[width=10cm]{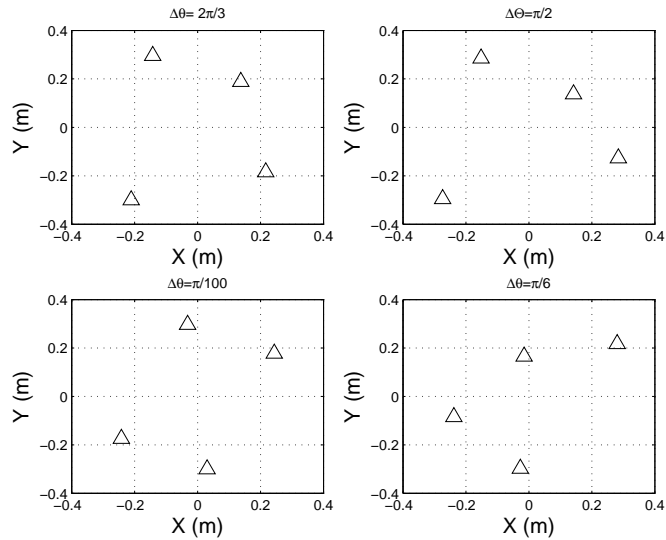}
\caption{The optimal antenna configuration for the two-target scenario. The optimal structure is found for different values of $\Delta \theta$.}
\label{OptimalStructureDeltaTheta}
\end{figure}

\begin{figure}
\centering
\includegraphics[width=10cm]{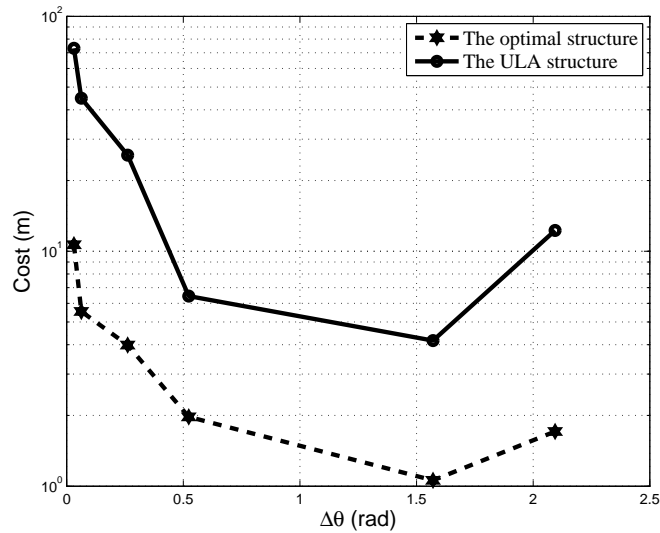}
\caption{The calculated cost for the two-target scenario. The cost was calculated for different values of $\Delta \theta$ and scenarios with the optimal and ULA structure.}
\label{CRLBDeltaTheta}
\end{figure}

\begin{figure}
\centering
\includegraphics[width=10cm]{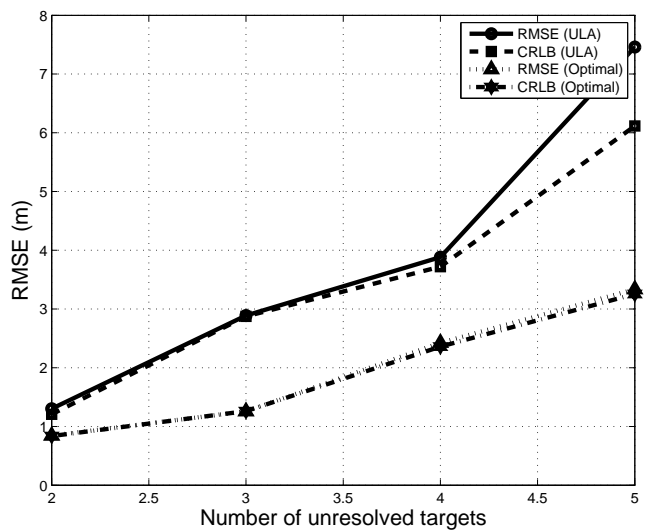}
\caption{Location RMSE and CRLB for both optimal and ULA structures. The results are obtained for scenarios with different number of targets being located inside the same resolution cell.}
\label{RMSECRLBULAOptimal}
\end{figure}
\renewcommand{\theequation}{\thesection.\arabic{equation}}
\appendices
\section{Proof of Proposition \ref{CRLBProp}}\label{ProofCRLB}
Consider the $n$-th target and the $m$-th target that are located in the $c_1$-th cell and the $c_2$-th cell, respectively. First, assume that $c_1=c_2-1$. The unknown term $J_{\theta^{c_1}_n\theta^{c_2}_m}$ is first calculated where the proof for other terms is similar. The entry $J_{\theta^{c_1}_n\theta^{c_2}_m}$ can be found using the following equality:
\begin{equation}\label{OriginalFIM}
J_{\theta^{c_1}_n\theta^{c_2}_m}=\left(\frac{\partial \bar{\rho}_*}{\partial \theta^{c_1}_n}\right)'\Sigma^{-1}_*\left(\frac{\partial \bar{\rho}_*}{\partial \theta^{c_2}_m}\right)+\mathcal{T}\left(
\left(\frac{\partial \Sigma_*}{\partial \theta^{c_1}_n}\right)\Sigma^{-1}_*\left(\frac{\partial \Sigma_*}{\partial \theta^{c_2}_m}\right)\Sigma^{-1}_*
\right)
\end{equation}
Based on the definition of the covariance matrix in (\ref{Covariance}), it is evident that the second term on the right-hand side of the above equation is zero. Now, according to the definition of $\bar{\rho}_*$ in (\ref{rostar}), the following equations can be derived for the partial derivative terms in (\ref{OriginalFIM}):
\begin{eqnarray}
\frac{\partial \bar{\rho}_*}{\partial \theta^{c_1}_n} &=& \left[\frac{\partial \Re(\bar{\eta}_{c_1-1})}{\partial \theta^{c_1}_n} \ \frac{\partial \Im(\bar{\eta}_{c_1-1})}{\partial \theta^{c_1}_n} \ \frac{\partial \Re(\bar{\eta}_{c_1})}{\partial \theta^{c_1}_n} \ \frac{\partial \Im(\bar{\eta}_{c_1})}{\partial \theta^{c_1}_n} \ 0\right]\\
\frac{\partial \bar{\rho}_*}{\partial \theta^{c_2}_m} &=& \left[0 \ \frac{\partial \Re(\bar{\eta}_{c_1})}{\partial \theta^{c_2}_m} \ \frac{\partial \Im(\bar{\eta}_{c_1})}{\partial \theta^{c_2}_m} \ \frac{\partial \Re(\bar{\eta}_{c_2})}{\partial \theta^{c_2}_m} \ \frac{\partial \Im(\bar{\eta}_{c_2})}{\partial \theta^{c_2}_m}\right]
\end{eqnarray}
After some algebraic operations, the FIM in (\ref{OriginalFIM}) can be written in the following form:
\begin{eqnarray}\label{Expansion}
J_{\theta^{c_1}_n\theta^{c_2}_m} &=& k_2\left\{\left(\frac{\partial \Re(\bar{\eta}_{c_1})}{\partial \theta^{c_1}_n}\right)\left(\frac{\partial \Re(\bar{\eta}_{c_1})}{\partial \theta^{c_2}_m}\right)'+ \left(\frac{\partial \Im(\bar{\eta}_{c_1})}{\partial \theta^{c_1}_n}\right)\left(\frac{\partial \Im(\bar{\eta}_{c_1})}{\partial \theta^{c_2}_m}\right)'\right\}\nonumber\\
&+& k_4\left\{\left(\frac{\partial \Re(\bar{\eta}_{c_1-1})}{\partial \theta^{c_1}_n}\right)\left(\frac{\partial \Re(\bar{\eta}_{c_1})}{\partial \theta^{c_2}_m}\right)'+ \left(\frac{\partial \Im(\bar{\eta}_{c_1-1})}{\partial \theta^{c_1}_n}\right)\left(\frac{\partial \Im(\bar{\eta}_{c_1})}{\partial \theta^{c_2}_m}\right)\right\}\nonumber \\
&+& k_5\left\{\left(\frac{\partial \Re(\bar{\eta}_{c_1-1})}{\partial \theta^{c_1}_n}\right)\left(\frac{\partial \Re(\bar{\eta}_{c_2})}{\partial \theta^{c_2}_m}\right)'+ \left(\frac{\partial \Im(\bar{\eta}_{c_1-1})}{\partial \theta^{c_1}_n}\right)\left(\frac{\partial \Im(\bar{\eta}_{c_2})}{\partial \theta^{c_2}_m}\right)'\right\} \nonumber\\
&+& k_6\left\{\left(\frac{\partial \Re(\bar{\eta}_{c_1})}{\partial \theta^{c_1}_n}\right)\left(\frac{\partial \Re(\bar{\eta}_{c_2})}{\partial \theta^{c_2}_m}\right)'+ \left(\frac{\partial \Im(\bar{\eta}_{c_1})}{\partial \theta^{c_1}_n}\right)\left(\frac{\partial \Im(\bar{\eta}_{c_2})}{\partial \theta^{c_2}_m}\right)'\right\}
\end{eqnarray}
Using (\ref{Average}) and (\ref{Psi}) and the notations given by definitions \ref{Notation1} and \ref{Notation2}, the following expressions can be derived for the derivatives in (\ref{Expansion}):
\begin{eqnarray}
\left(\frac{\partial \Re(\bar{\eta}_{c_1})}{\partial \theta^{c_1}_n}\right)\left(\frac{\partial \Re(\bar{\eta}_{c_1})}{\partial \theta^{c_2}_m}\right)' = K\left(\frac{2\pi}{\lambda}\right)^2\beta^{c_1}_n(1-\beta^{c_2}_m)\times \nonumber\\
\left\{\sum_{l=1}^{nm}(\textbf{p}^{c_1}_n)'\Omega(:,l)\Omega'(:,l)\textbf{p}^{c_2}_m
\left(\bar{\xi}^{c_1}_n\cos(\omega^{c_1}_n(l))-\bar{\zeta}^{c_1}_n\sin(\omega^{c_1}_n(l))\right)\left(\bar{\xi}^{c_2}_m\cos(\omega^{c_2}_m(l))-\bar{\zeta}^{c_2}_m\sin(\omega^{c_2}_m(l))\right)\right\} \nonumber \\
\left(\frac{\partial \Im(\bar{\eta}_{c_1})}{\partial \theta^{c_1}_n}\right)\left(\frac{\partial \Im(\bar{\eta}_{c_1})}{\partial \theta^{c_2}_m}\right)'= K\left(\frac{2\pi}{\lambda}\right)^2\beta^{c_1}_n(1-\beta^{c_2}_m)\times \nonumber\\
\left\{\sum_{l=1}^{nm}(\textbf{p}^{c_1}_n)'\Omega(:,l)\Omega'(:,l)\textbf{p}^{c_2}_m \left(\bar{\xi}^{c_1}_n\sin(\omega^{c_1}_n(l))+\bar{\zeta}^{c_1}_n\cos(\omega^{c_1}_n(l))\right)\left(\bar{\xi}^{c_2}_m\sin(\omega^{c_2}_m(l))+\bar{\zeta}^{c_2}_m\cos(\omega^{c_2}_m(l))\right)\right\}
\end{eqnarray}
Therefore, the first term on the right-hand side of (\ref{Expansion}) can be written as
\begin{eqnarray}
k_2\left\{\left(\frac{\partial \Re(\bar{\eta}_{c_1})}{\partial \theta^{c_1}_n}\right)\left(\frac{\partial \Re(\bar{\eta}_{c_1})}{\partial \theta^{c_2}_m}\right)'+ \left(\frac{\partial \Im(\bar{\eta}_{c_1})}{\partial \theta^{c_1}_n}\right)\left(\frac{\partial \Im(\bar{\eta}_{c_1})}{\partial \theta^{c_2}_m}\right)'\right\} &=& \nonumber\\
K\left(\frac{2\pi}{\lambda}\right)^2 \beta^{c_1}_n(1-\beta^{c_2}_m)\left\{\sum_{l=1}^{nm} (\textbf{p}^{c_1}_n)'\Omega(:,l)\Omega'(:,l)\textbf{p}^{c_2}_m \left(\kappa^{nm}_{c_1c_2}\cos(\omega^{c_2}_m(l)-\omega^{c_1}_n(l))+\iota^{nm}_{c_1c_2}\sin(\omega^{c_2}_m(l)-\omega^{c_1}_n(l))\right)\right\}
\end{eqnarray}
Other terms on the right-hand side of (\ref{Expansion}) can be similarly found. The final form can be now written as follows:
\begin{eqnarray}
J_{\theta^{c_1}_n\theta^{c_2}_m} = K\left(\frac{2\pi}{\lambda}\right)^2\left\{\sum_{l=1}^{nm} (\textbf{p}^{c_1}_n)'\Omega(:,l)\Omega'(:,l)\textbf{p}^{c_2}_m \left(\kappa^{nm}_{c_1c_2}\cos(\omega^{c_2}_m(l)-\omega^{c_1}_n(l))+\iota^{nm}_{c_1c_2}\sin(\omega^{c_2}_m(l)-\omega^{c_1}_n(l))\right)\right\} \times \nonumber \\
\left(k_2 \beta^{c_1}_n(1-\beta^{c_2}_m)+k_4 (1-\beta^{c_1}_n)(1-\beta^{c_2}_m)+k_5 \beta^{c_1}_n\beta^{c_2}_m+k_6 \beta^{c_1}_n(1-\beta^{c_2}_m) \right)
\end{eqnarray}
where the term in the second line of the above equation can be written as follows:
\begin{eqnarray}
\left(k_2 \beta^{c_1}_n(1-\beta^{c_2}_m)+k_4 (1-\beta^{c_1}_n)(1-\beta^{c_2}_m)+k_5 \beta^{c_1}_n\beta^{c_2}_m+k_6 \beta^{c_1}_n(1-\beta^{c_2}_m) \right) =
\left[(\breve{\beta}^{c_1}_n)' \ 0\right]\Sigma^{-1}_*\left[0 \ (\breve{\beta}^{c_2}_m)'\right]'
\end{eqnarray}
which is the coefficient $C_{\theta^{c_1}_n\theta^{c_2}_m}$. For the case with $c_1=c_2$, the same procedure can be followed and the expression in the proposition is similarly found.
\section{Proof of Theorem \ref{Theorem1}}\label{ProofTheorem1}
We begin with the optimization formulation given by (\ref{Optimization3}). Define the new matrix $T_{nm}$ and  the new variable $t_{nm}$ with $\{m=1,\cdots,M\}, \{n=1,\cdots,N\}$, and rewrite the optimization problem as follows:
\begin{equation}\label{Optimization4}
\begin{array}{cc}
\max_{\left\{\Delta \textbf{s}_{11},\cdots,\Delta \textbf{s}_{nm},\textbf{t},T^*\right\}} & \sum_{m=1}^M\sum_{n=1}^N t_{nm}\\
\mbox{S.T} & ||\Delta \textbf{s}_{nm}||_2 \geq d_{nm}\\
 & ||\Delta \textbf{s}_{nm}||_2 \leq e_{nm}\\
 & \sum_{m=1}^M \textbf{s}_{tm}+\sum_{n=1}^N \textbf{s}_{rn}=0\\
 & (\textbf{p}^c)'T_{nm}\textbf{p}^c \geq t_{nm}\\
 & (\Delta \textbf{s}_{nm})(\Delta \textbf{s}_{nm})' \succeq T_{nm}, \forall \ m=\{1,\cdots,M\}, n=\{1,\cdots,N\}
\end{array}
\end{equation}
where $\textbf{t}=\left[t_{11} \ \cdots \ t_{nm}\right]'$, and $T^*=\{T_{11},\cdots,T_{nm}\}$. The second-norm terms in the constraints can be written into the following form:
\begin{eqnarray}\label{Ineq1}
\left[\begin{array}{cc}
-I_{2\times2} & \Delta \textbf{s}_{nm}\\
\Delta \textbf{s}'_{nm} & -e^{2}_{nm}
\end{array}
\right] \preceq 0, \ \left[\begin{array}{cc}
I_{2\times2} & \Delta \textbf{s}_{nm}\\
\Delta \textbf{s}'_{nm} & d^{2}_{nm}
\end{array}
\right] \preceq 0
\end{eqnarray}
In addition, using the Schur-complement of a square matrix, the last constraint in (\ref{Optimization4}) is written as
\begin{equation}\label{Ineq2}
\left[\begin{array}{cc}
1 & \Delta \textbf{s}'_{nm}\\
\Delta \textbf{s}_{nm} & T_{nm}
\end{array}
\right] \preceq 0
\end{equation}
Inserting the new forms provided by (\ref{Ineq1}) and (\ref{Ineq2}) in (\ref{Optimization4}) and using the fact that $(\textbf{p}^c)'T_{nm}\textbf{p}^c=\mathcal{T}(T_{nm}P)$ with $P=\textbf{p}^c(\textbf{p}^c)'$, the following new form is derived for the optimization problem:
\begin{equation}\label{Optimization5}
\begin{array}{cc}
\max_{\left\{T^*,S^*,\textbf{t}\right\}} & \sum_{m=1}^M\sum_{n=1}^N t_{nm}\\
\mbox{S.T} & \sum_{m=1}^M \textbf{s}_{tm}+\sum_{n=1}^N \textbf{s}_{rn}=0\\
 & \mathcal{T}(T_{nm}P) \geq t_{nm}\\
 & \left[\begin{array}{cc}
-I_{2\times2} & \Delta \textbf{s}_{nm}\\
\Delta \textbf{s}'_{nm} & -e^{2}_{nm}
\end{array}
\right] \preceq 0, \ \left[\begin{array}{cc}
I_{2\times2} & \Delta \textbf{s}_{nm}\\
\Delta \textbf{s}'_{nm} & d^{2}_{nm}
\end{array}
\right] \preceq 0\\
 & \nonumber\\
 & \left[\begin{array}{cc}
1 & \Delta \textbf{s}'_{nm}\\
\Delta \textbf{s}_{nm} & T_{nm}
\end{array}
\right] \preceq 0, \forall \ m=\{1,\cdots,M\}, n=\{1,\cdots,N\}
\end{array}
\end{equation}
where $S^*=\{\textbf{s}_{t1},\cdots,\textbf{s}_{rN}\}$. Now, the difference vector is written as
\begin{equation}
\Delta \textbf{s}_{nm}=\textbf{s}_{rn}-\textbf{s}_{tm}
\end{equation}
Replacing the above equation in (\ref{Optimization5}), the form given by the theorem is obtained.
\section{Proof of Lemma \ref{RotationLemma}}\label{Rotation}
Consider the optimal structure found for the case with $\theta_1$. Assuming $\theta_2$ as the new DOA the cost function in (\ref{Optimization3}) can be rewritten as follows:
\begin{equation}
\sum_{m=1}^M\sum_{n=1}^N \Delta \textbf{s}_{nm}'P^* \Delta \textbf{s}_{nm}
\end{equation}
with $P^*=(\textbf{p}^*)'\textbf{p}^*$ and $\textbf{p}^*=\left[\cos(\theta_1+\Delta \theta) \ -\sin(\theta_1+\Delta \theta)\right]'$. The vector $\textbf{p}^*$ can be expanded as
\begin{equation}
\textbf{p}^* = \left[\begin{array}{cc}
\cos(\Delta \theta) & \sin(\Delta \theta)\\
-\sin(\Delta \theta) & \cos(\Delta \theta)
\end{array}
\right]\left[\begin{array}{c}
\cos(\theta_1)\\
-\sin(\theta_1)
\end{array}
\right]
\end{equation}
Defining the first term on the right-hand side of the above equation as $G_{\Delta \theta}$, the cost function can be rewritten as follows:
\begin{equation}
\sum_{m=1}^M\sum_{n=1}^N \Delta \textbf{s}'_{nm}G_{\Delta \theta}PG'_{\Delta \theta} \Delta \textbf{s}_{nm}
\end{equation}
We know that $\Delta \textbf{s}^{o1}_{nm}=\left(\textbf{s}^{o1}_{tm}-\textbf{s}^{o1}_{rn}\right)$ maximizes the cost function in (\ref{Optimization3}) where $P$ is the matrix corresponding the target with $\theta_1$ as the DOA. Therefore, an optimal solution of the optimization problem with $\theta_2$ as the DOA of the target can be obtained as
\begin{eqnarray}
G'_{\Delta \theta}\textbf{s}^{o2}_{tm} &=& \textbf{s}^{o1}_{tm}\\
G'_{\Delta \theta}\textbf{s}^{o2}_{rn} &=& \textbf{s}^{o1}_{rn}, \ \forall \ m=\{1,\cdots,M\}, n=\{1,\cdots,N\}
\end{eqnarray}
Consequently, the new optimal solution is written as follows:
\begin{eqnarray}
\textbf{s}^{o2}_{tm} &=& G_{\Delta \theta}\textbf{s}^{o1}_{tm}\\
\textbf{s}^{o2}_{rn} &=& G_{\Delta \theta}\textbf{s}^{o1}_{rn}, \ \forall \ m=\{1,\cdots,M\}, n=\{1,\cdots,N\}
\end{eqnarray}
Now, we have to check whether the new solution holds in the constraints. It can be shown that:
\begin{equation}
||\Delta \textbf{s}^{o2}_{nm}||_2=||\Delta \textbf{s}^{o1}_{nm}||_2
\end{equation}
In addition, it is known that:
\begin{equation}
\sum_{m=1}^M \textbf{s}^{o2}_{tm}+\sum_{n=1}^N \textbf{s}^{o2}_{rn} = G_{\Delta \theta}\left(\sum_{m=1}^M \textbf{s}^{o1}_{tm}+\sum_{n=1}^N \right) = 0
\end{equation}
which implies that the new optimal solution also meets the constraints.
\section{Proof of Proposition \ref{Uniqueness}}\label{UniquenessProof}
Consider the optimization problem in (\ref{Optimization3}) without the constraint on the mass center. In this case, the cost function is quadratic with respect to the unknown difference vectors. The unique optimal solution obtained by solving the resulting optimization problem can be written as $\{\Delta \textbf{s}^o_{nm}\}$ with $m=\{1,\cdots,M\}$ and $n=\{1,\cdots,N\}$. It is evident that there are an infinite number of location solutions for which the above set of difference vectors are obtained. Let us define the $i$-th and the $j$-th sets as $\{S^{oi}_t,S^{oi}_r\}$ and $\{S^{oj}_t,S^{oj}_r\}$, respectively. It is known from the geometry that:
\begin{equation}\label{Transform}
\textbf{s}^{oi}_{tl}=G_{\theta}\textbf{s}^{oj}_{tl}+\textbf{b}_{tl}
\end{equation}
where $G_{\theta}$ denotes a rotation matrix with $\theta$ as the angle of the rotation, and $\textbf{b}$ refers to an arbitrary translation. Note that the above equation can be  written for every other antenna in the array of receivers as well. Considering the mass center constraint given by (\ref{CenterMass}), we show that the translation should be zero in (\ref{Transform}). To show this, we first assume that there is a nonzero translation as $\textbf{b}_{tl}$. Then, it is observed that such an assumption leads to the contradiction. From the assumption \ref{CenterMass}, it is known that the center of the mass of the array is located in the origin. Therefore, there should be another translation $\textbf{b}_{tv}$ where $\textbf{b}_{tv}=-\textbf{b}_{tl}$. Under the new translations, the new difference vector is written as
\begin{equation}
(\textbf{s}^{oi}_{tl}-\textbf{s}^{oi}_{tv}) = (\textbf{s}^{oj}_{tl}-\textbf{s}^{oj}_{tv}+2\textbf{b}_{tl})
\end{equation}
It is now evident that the new configuration gives a different set of difference vectors, which is a contradiction to our initial assumption (e.g., the same set of difference vectors). Therefore, the translation part in (\ref{Transform}) is zero. Now, consider the rotation part in (\ref{Transform}). It is known that the rotation transform does not change the distance between each two points. Rewrite the cost function in (\ref{Optimization3}) into the following form:
\begin{equation}
\sum_{m=1}^{N}\sum_{n=1}^N (\textbf{s}^{oi}_{nm})'P^c \Delta \textbf{s}^{oi}_{nm}= \sum_{m=1}^{N}\sum_{n=1}^N (\textbf{s}^{oj}_{nm})'G'_{\theta}P^c G_{\theta} \Delta \textbf{s}^{oj}_{nm}
\end{equation}
Using formal matrix operations, the new matrix $U=G'_{\theta}P^c G_{\theta}$ can be written in the following form:
\begin{equation}\label{U}
U=\left[\begin{array}{cc}
\cos^2(\theta+\theta^c) & -\sin(\theta+\theta^c)\cos(\theta+\theta^c)\\
-\sin(\theta+\theta^c)\cos(\theta+\theta^c) & \sin^2(\theta+\theta^c)
\end{array}
\right]
\end{equation}
The cost produced by each of two sets of optimal solutions is equal if the following condition is held:
\begin{equation}\label{ZeroEquality}
\sum_{m=1}^M\sum_{n=1}^N (\Delta \textbf{s}^{oi}_{nm})'V\Delta \textbf{s}^{oi}_{nm}=0
\end{equation}
with $V=U-P^c$. The equality in (\ref{ZeroEquality}) is valid if either $V=0$ or $V$ is neither positive nor negative semi-definite. First, assume $V=0$. Based on the given form in (\ref{U}) for the matrix $U$, it can be inferred that $U=P^c$ when $\theta=n \pi$. In other words, a rotation with $n \pi$ as the angle of rotation provides the same cost function. Now, assume the other case where $V \neq 0$. It can be shown that matrix $V$ has two eigenvalues $\{\lambda,-\lambda\}$ where the value of $\lambda$ depends on the rotation angle and $\theta^c$. Therefore, the zero inequality in (\ref{ZeroEquality}) leads to a number of solutions for the difference vectors. The rotated configuration can be then another solution of the optimization problem if $\{S^{oi}_t,S^{oi}_r\}$ belongs to the set of solutions of (\ref{ZeroEquality}). The above discussions state that the optimization problem provides at least two solutions for the optimum configuration of antennas.
\end{document}